\documentclass[11pt]{article}

\usepackage{amsmath}
\usepackage{amssymb}
\usepackage{amsthm}

\usepackage{algorithm}
\usepackage{algorithmicx}
\usepackage{algpseudocode}

\usepackage{graphicx}
\usepackage{subcaption}

\usepackage{hyperref}

\usepackage{color}
\usepackage[top=1in, bottom=1in, left=1in, right=1in]{geometry}

\usepackage{todonotes}

\algblockdefx[Choose]{Choose}{EndChoose}%
[1]{\textbf{choose} \hbox{#1}}{\textbf{end choose}}

\newtheorem{theorem}{Theorem}

\newtheorem{lemma}[theorem]{Lemma}
\newtheorem{claim}[theorem]{Claim}

\newtheorem{conjecture}[theorem]{Conjecture}
\newtheorem{question}[theorem]{Question}
\theoremstyle{definition}
\newtheorem{definition}[theorem]{Definition}

\def\pprime{{\prime\prime}}
\def\ppprime{{\prime\prime\prime}}

\def\Merge{{\textit{Merge}}}
\def\Rtim{{{\mathcal R}_{\mathrm{tim}}}}
\def\Rastack{{{\mathcal R}_{\alpha\mathrm{st}}}}
\def\Rshiver{{{\mathcal R}_{\mathrm{sh}}}}
\def\Ramerge{{{\mathcal R}_{\alpha\mathrm{m}}}}
\def\YZ{{\mathit{YZ}}}

\def\QQ{Q}

\title{
    Strategies for Stable Merge Sorting
}

\author{
  Sam Buss \\
  {\small \texttt{sbuss@ucsd.edu}} \\
  {\small Department of Mathematics} \\
  {\small University of California, San Diego} \\
  {\small La Jolla, CA, USA}
  \and
  Alexander Knop \\
  {\small \texttt{aknop@ucsd.edu}} \\
  {\small Department of Mathematics} \\
  {\small University of California, San Diego} \\
  {\small La Jolla, CA, USA}
}

\begin{document}

\maketitle

\begin{abstract}
We introduce new stable natural merge sort algorithms, called $2$-merge
sort and $\alpha$-merge sort. We prove upper and lower bounds for several
merge sort algorithms, including Timsort, Shivers' sort, $\alpha$-stack
sorts, and our new $2$-merge and $\alpha$-merge sorts. The upper and lower
bounds have the forms $c \cdot n \log m$ and $c \cdot n \log n$ for inputs
of length~$n$ comprising $m$~monotone runs. For Timsort, we prove a lower bound of
$(1.5 - o(1)) n \log n$. For $2$-merge sort, we prove optimal upper and
lower bounds of approximately $(1.089 \pm o(1))n \log m$. We prove similar
asymptotically matching upper and lower bounds for $\alpha$-merge sort, when
$\varphi < \alpha < 2$, where $\varphi$ is the golden ratio.

Our bounds are in terms of merge cost; this upper bounds the number of comparisons
and accurately models runtime.
The merge strategies can be used for any stable merge sort, not just
natural merge sorts. The new $2$-merge and $\alpha$-merge sorts have better
worst-case merge cost upper bounds and are slightly simpler to implement than
the widely-used Timsort; they also perform better in experiments. We report
also experimental comparisons with algorithms developed by Munro-Wild and Jug\'e
subsequently to the results of the present paper.
\end{abstract}

\stepcounter{footnote}
\footnotetext{The present paper is an expansion of an earlier conference
paper~\cite{BussKnop:StableMergeSortSODA}. It includes the contents
of~\cite{BussKnop:StableMergeSortSODA} plus a new easy, direct proof of
an $O(n \log m)$ upper bound
for $\alpha$-merge sort (Theorem~\ref{thm:alphaMergeUpperEasy}),
the proof of the tight upper bounds for $\alpha$-merge sort
(Theorem~\ref{thm:alphamergeUpper}), and experimental comparisons with
the algorithms of
Munro-Wild~\cite{MunroWild:MergeSort}
and Jug\'e~\cite{Juge:AdaptiveShivers}
that were developed after this paper was first circulated.}

\section{Introduction}
This paper studies stable merge sort algorithms, especially natural merge sorts.
We will propose new strategies for the order in which merges are performed, and
prove upper and lower bounds on the cost of several merge strategies. The first
merge sort algorithm was proposed by von
Neumann~\cite[p.159]{Knuth:SortingSearching}: it works by splitting the input
list into sorted sublists, initially possibly lists of length one, and then
iteratively merging pairs of sorted lists, until the entire input is sorted. A
sorting algorithm is {\em stable} if it preserves the relative order of elements
which are not distinguished by the sort order. There are several methods of
splitting the input into sorted sublists before starting the merging; a merge
sort is called {\em natural} it finds the sorted sublists by detecting
consecutive runs of entries in the input which are already in sorted order.
Natural merge sorts were first proposed by Knuth~\cite[p.160]{Knuth:SortingSearching}.

Like most sorting algorithms, the merge sort is comparison-based in that
it works by comparing the relative order of pairs of entries in the input list.
Information-theoretic considerations imply
that any comparison-based sorting algorithm must make at least
$\log_2(n!) \approx  n \log_2 n$ comparisons in the worst case.
However, in many practical applications, the input is frequently already partially
sorted.  There are many {\em adaptive} sort algorithms which will
detect this and run faster on inputs which are already partially sorted.
Natural merge sorts are adaptive in this sense: they detect sorted
sublists (called ``runs'') in the input, and thereby reduce the
cost of merging sublists. One very popular stable natural merge sort is the
eponymous Timsort of Tim Peters~\cite{Peters:TimSortPost}.
Timsort is extensively used, as it is included in Python, in the Java standard
library, in GNU Octave, and in the Android operating system. Timsort has
worst-case runtime $O(n \log n)$, but is designed to run substantially faster on
inputs which are partially pre-sorted by using intelligent strategies to
determine the order in which merges are performed.

There is extensive literature on adaptive sorts: e.g., for theoretical
foundations see \cite{Mannila:Presortedness,%
EstivillCastroWood:surveyAdaptive,%
PeterssonMoffat:OverviewAdaptive,%
McIlroy:OptimisticSorting} and
for more applied investigations see
\cite{%
CookKim:NearlySorted,%
Harris:NaturalMerge,%
ChandramouliGoldstein:Patience,%
ZML:SortRace,%
LaRoccaCantone:neatsort}.
The present paper will consider only stable, natural merge sorts.  As
exemplified by the wide deployment of Timsort, these are certainly an important
class of adaptive sorts. We will consider
the Timsort algorithm~\cite{Peters:TimSortPost,GRBBH:TimSortBug}, and
related sorts due to Shivers~\cite{Shivers:mergesort}
and Auger-Nicaud-Pivoteau~\cite{ANP:MergeStrategies}. We will
also introduce new algorithms,
the ``2-merge sort'' and
the ``$\alpha$-merge sort'' for $\varphi < \alpha < 2$ where
$\varphi$ is the golden ratio.

The main contribution of the present paper is the definition
of new stable merge sort algorithms, called $2$-merge and
$\alpha$-merge sort. These have better worst-case run times
than Timsort, are slightly easier to implement
than Timsort, and perform better in our experiments.

We focus on natural merge sorts, since they are so widely used. However, our
central contribution is analyzing merge strategies and our results are
applicable to any stable sorting algorithm that generates and merges runs,
including patience sorting~\cite{ChandramouliGoldstein:Patience},
melsort~\cite{Skiena:Encroaching,LevcopoulosPetersson:SortingShuffled},
and split sorting~\cite{LevcopoulosPetersson:SplitSort}.

All the merge sorts we consider will use the following framework. (See
Algorithm~\ref{alg:basic}.) The input is a list of $n$ elements.
The first logical stage of the
algorithm (following~\cite{Peters:TimSortPost}) identifies maximal
length subsequences of consecutive entries which are in sorted order, either
ascending or descending. The descending subsequences are reversed, and this
partitions the input into ``runs'' $R_1, \ldots, R_m$ of entries sorted in
non-decreasing order. The number of runs is~$m$; the number of elements in a
run~$R$ is $|R|$. Thus $\sum_i |R_i| = n$. It is easy to see that these runs may
be formed in linear time and with linear number of comparisons.

Merge sort algorithms process the runs in left-to-right order starting
with~$R_1$. This permits runs to be identified on-the-fly,
only when needed. This means there is no need to allocate $\Theta(m)$ additional
memory to store the runs. This also may help reduce cache
misses. On the other hand, it means that the value $m$
is not known until the final run is formed; thus, the natural
sort algorithms do not use $m$ except
as a stopping condition.

The runs~$R_i$ are called {\em original runs}. The second logical stage of the
natural merge sort algorithms repeatedly merges runs in pairs to give longer and
longer runs. (As already alluded to, the first and second logical stages are
interleaved in practice.) Two runs $A$ and $B$ can be merged in linear time;
indeed with only $|A| + |B| - 1$ many comparisons and $|A| + |B|$ many movements
of elements. The merge sort stops when all original runs have been identified
and merged into a single run.

Our mathematical model for the run time of a merge sort algorithm is the sum,
over all merges of pairs of runs $A$ and~$B$, of $|A| + |B|$.  We call the
quantity the {\em merge cost}. In most situations, the run time of a
natural sort algorithm can be linearly bounded in terms of its merge cost. Our
main theorems are lower and upper bounds on the merge cost of several stable
natural merge sorts.  Note that if the runs are merged in a balanced fashion,
using a binary tree of height $\lceil{\log m}\rceil$, then the total merge cost
$ \le n \lceil{\log m}\rceil$. (We use $\log$ to denote logarithms base~2.)
Using a balanced binary tree of merges gives  a good worst-case merge cost, but
it does not take into account savings that are available when runs have
different lengths.\footnote{%
  \cite{HMW:Presorted} gives a different method of achieving merge cost
  $O( n \log m )$. Like the binary tree method, their method is not adaptive.
}
The goal is find {\em adaptive} stable natural merge sorts which can
take advantage of different run lengths to reduce the merge cost even below $n\log m$,
but which are
guaranteed to never be much worse than the binary tree. Therefore, our preferred
upper bounds on merge costs are stated in the form $c \cdot n \log m$ for some
constant~$c$, rather than in the form $c\cdot n \log n$.

The merge cost ignores the $O(n)$ cost of forming the original runs~$R_i$: this
does not affect the asymptotic limit of the constant~$c$.

Algorithm~\ref{alg:basic} shows the framework for all the merge sort algorithms
we discuss. This is similar to what Auger et al.~\cite{ANP:MergeStrategies} call
the ``generic'' algorithm. The input~$S$ is a sequence of integers which is
partitioned into monotone runs of consecutive members. The decreasing runs are
inverted, so $S$ is expressed as a list $\mathcal R$ of increasing original runs
$R_1, \ldots, R_m$ called ``original runs''. The algorithm maintains a
stack~$\mathcal \QQ$ of runs $\QQ_1, \ldots, \QQ_\ell$, which have been formed from
$R_1, \ldots, R_k$. Each time through the loop, it either pushes
the next original run, $R_{k + 1}$, onto the stack $\mathcal \QQ$, or it chooses a
pair of adjacent runs $\QQ_i$ and~$\QQ_{i+1}$ on the stack and merges them. The
resulting run replaces $\QQ_i$ and~$\QQ_{i+1}$ and becomes the new $\QQ_i$, and the
length of the stack decreases by one. The entries of the runs $\QQ_i$ are stored
in-place, overwriting the elements of the array which held the input~$S$.
Therefore, the stack needs to hold only the positions~$p_i$ (for $i = 0, \ldots,
\ell +  1$) in the input array where the runs $\QQ_i$ start, and thereby implicitly
the lengths $|\QQ_i|$ of the runs. We have $p_1=0$, pointing to the beginning of
the input array, and for each~$i$, we have $|\QQ_i| = p_{i + 1} - p_i$. The
unprocessed part of~$S$ in the input array starts at position $p_{\ell + 1} =
p_\ell + |\QQ_\ell|$. If $p_{\ell + 1} < n$, then it will be the starting position
of the next original run pushed onto~$\mathcal \QQ$.

\begin{algorithm}
    \begin{algorithmic}[1]
      \Procedure{MergeSortFramework}{$S$, $n$}
        \State $\mathcal{R} \gets \text{list of runs forming $S$}$
        \State $\mathcal \QQ \gets \text{empty stack}$
        \While{$\mathcal{R} \neq \emptyset$ or $\mathcal \QQ$ has $>1$ member}
          \Choose {to do either {\bf (A)} or
                  {\bf (B${}_i$)} for some $\ell  -  k_2 < i < \ell$}, based on
            \State whether $\mathcal R$ is empty and
                     on the values $|\QQ_j|$ for $\ell  -  k_1 < j \le \ell$
    \def\ChoicesIndent{2.2em}
            \State \hbox to\ChoicesIndent{\bf (A)}%
                    Remove the next run $R$ from $\mathcal{R}$
                    and push it onto $\mathcal \QQ$.
            \State \kern\ChoicesIndent This increments $\ell = |\mathcal \QQ|$ by 1.
            \State \hbox to\ChoicesIndent{\bf (B${}_i$)}%
                    Replace $\QQ_i$ and $\QQ_{i+1}$ in $\mathcal \QQ$ with
                    $\Merge(\QQ_i, \QQ_{i+1})$.
            \State \kern\ChoicesIndent This decrements $\ell = |\mathcal \QQ|$ by 1.
          \EndChoose
        \EndWhile
        \State Return $\QQ_1$ as the sorted version of~$S$.
      \EndProcedure
    \end{algorithmic}
    \caption{
        The basic framework for all our merge sort algorithms.
        \newline $S$ is a
        sequence of integers of length~$n$. $\mathcal{R}$~is the list of $m$
        runs formed from~$S$.
        \newline
        $\mathcal \QQ$~is a stack of $\ell$~runs,
        $\QQ_1, \QQ_2, \ldots, \QQ_\ell$. The top member of the
        stack~$\mathcal \QQ$ is $\QQ_\ell$.
        \newline
        $k_1$ and $k_2$ are fixed (small)
        integers. The algorithm is $(k_1,k_2)$-aware.
        \newline
        Upon termination, $\mathcal \QQ$ contains a single run~$\QQ_1$
        which is the sorted version of~$S$.
    }
    \label{alg:basic}
\end{algorithm}

Algorithm~\ref{alg:basic} is called {\em $(k_1,k_2)$-aware} since its choice of
what to do is based on just the lengths of the runs in the top $k_1$ members of
the stack~$\mathcal \QQ$, and since merges are only applied to runs in the top
$k_2$ members of~$\mathcal \QQ$. (\cite{ANP:MergeStrategies} used the terminology
``degree'' instead of ``aware''.) In all our applications, $k_1$ and $k_2$ are
small numbers, so it is appropriate to store the runs in a stack. Usually
$k_1=k_2$, and we write ``$k$-aware'' instead of ``$(k,k)$-aware''.
Table~\ref{tab:awareness} shows the awareness values for the algorithms
considered in this paper. To improve
readability (and following \cite{ANP:MergeStrategies}), we use the letters
$W,X,Y,Z$ to denote the top four runs on the stack,
$\QQ_{\ell - 3}, \QQ_{\ell - 2}, \QQ_{\ell - 1}, \QQ_\ell$ respectively, (if they
exist).\footnote{\cite{Peters:TimSortPost} used ``$A,B,C$'' for ``$X,Y,Z$''.}

\begin{table}[ht]
    \begin{center}
        \begin{tabular}{cc}
            Algorithm & Awareness \\
            \hline
            Timsort (original) \cite{Peters:TimSortPost} & 3-aware \\
            Timsort (corrected) \cite{GRBBH:TimSortBug} & (4,3)-aware \\
            $\alpha$-stack sort \cite{ANP:MergeStrategies} & 2-aware \\
            Shivers sort \cite{Shivers:mergesort} & 2-aware \\
            $2$-merge sort & 3-aware \\
            $\alpha$-merge sort ($\varphi{<}\alpha{<}2$) & 3-aware
        \end{tabular}
    \end{center}
    \caption{Awareness levels for merge sort algorithms}
    \label{tab:awareness}
\end{table}

In all the sorting algorithms we consider, the height of the stack~$\mathcal \QQ$
will be small, namely $\ell = O(\log n)$. Since the stack needs only store the
$\ell +  1$ values $p_1, \ldots, p_{\ell + 1}$, the memory requirements for the stack
are minimal. Another advantage of Algorithm~\ref{alg:basic} is that runs
may be identified on the fly and that merges occur only near the top of the
stack: this may help reduce cache misses.
(See~\cite{LaMarcaLadner:CacheSorting} for other methods for reducing cache
misses.)

$\Merge(A,B)$ is the run obtained by stably merging $A$ and~$B$. Since it takes
only linear time to extract the runs $\cal R$ from~$S$, the computation time of
Algorithm~\ref{alg:basic} is dominated by the time needed for merging runs. As
already mentioned, we use $|A| + |B|$ as the mathematical model for the runtime
of a (stable)
merge. The usual algorithm for merging $A$ and $B$ uses an auxiliary buffer to
hold the smaller of $A$ and $B$, and then merging directly into the combined
buffer: this has runtime proportional to $|A| + |B|$.
Timsort~\cite{Peters:TimSortPost} uses a variety of techniques to speed
up merges, in particular ``galloping''; this also still takes time proportional
to $|A| + |B|$ in general. It is possible to perform (stable) merges in-place
with no additional memory~\cite{%
Kronrod:inplacemerging,
MannilaUkkonen:inplacemerging,
HuangLangston:inplacemerging,
Symvonis:inplacemerging,
GKP:inplacemerging}; these algorithms also require time $\Theta(|A| + |B|)$.

More complicated data structures can perform merges in sublinear time in some
situations; see for instance~\cite{CLP:SublinearMerging,HMW:Presorted}. These
methods do not seem to be useful in practical applications, and of course still
have worst-case run time $\Theta(|A| + |B|)$.

If $A$ and $B$ have very unequal lengths, there are (stable) merge algorithms
which use fewer than $\Theta(|A| + |B|)$
comparisons~\cite{HwangLin:merging,Symvonis:inplacemerging,GKP:inplacemerging}.
Namely, if $|A| \le |B|$, then it is possible to merge
$A$ and $B$ in time $O(|A| +  |B|)$ but using only $O(|A| (1 + \log(|B| / |A|)))$
comparisons.
Nonetheless, we feel that the cost $|A| + |B|$
is the best way to model the runtime of merge sort algorithms. Indeed,
the main strategies for speeding up merge sort algorithms try
to merge runs of approximate equal length as much as possible;
thus $|A|$ and $|B|$
are very unequal only in special cases. Of course, all our upper bounds
on merge cost are also upper bounds on number of comparisons.

\begin{definition}
    The {\em merge cost} of a merge sort algorithm on an input~$S$ is the sum of
    $|A| + |B|$ taken over all merge operations $\Merge(A,B)$ performed.  For
    $\QQ_i$ a run on the stack~$\mathcal \QQ$ during the computation, the {\em merge
    cost} of $\QQ_i$ is the sum of $|A| + |B|$ taken over all merges $\Merge(A,B)$
    used to combine runs that form~$\QQ_i$.
\end{definition}

Our definition of merge cost is motivated primarily by analyzing the run time
of {\em sequential} algorithms. Nonetheless, the constructions may be applicable
to distributed sorting algorithms, such as in
MapReduce~\cite{DeanGhemawat:mapreduce}. A distributed sorting algorithm
typically performs sorts on subsets of the input on multiple processors: each
of these can take advantage of a faster sequential merge sort algorithm.
Applications to distributed sorting are beyond the scope of the present paper
however.

In later sections, the notation $w_{\QQ_i}$ is used to denote
the merge cost of the $i$-th entry $\QQ_i$ on the stack at a given
time. Here ``$w$'' stands for ``weight''.
The notation $|\QQ_i|$ denotes the length of the run~$\QQ_i$.
We use $m_{\QQ_i}$ to denote the number of original runs
which were merged to form~$\QQ_i$.

All of the optimizations used by Timsort
mentioned above can be used equally well
with any of the merge sort algorithms discussed in the present paper.
In addition, they can be used with other sorting algorithms that
generate and merge runs. These other algorithms include patience
sort~\cite{ChandramouliGoldstein:Patience}, melsort~\cite{Skiena:Encroaching}
(which is an extension of patience sort), the hybrid quicksort-melsort
algorithm of~\cite{LevcopoulosPetersson:SortingShuffled}, and split
sort~\cite{LevcopoulosPetersson:SplitSort}.  The merge cost as defined above
applies equally well to all these algorithms. Thus, it gives a runtime
measurement which applies to a broad range of sort algorithms that incorporate
merges and which is largely independent of which optimizations are used.

Algorithm~\ref{alg:basic}, like all the algorithms we discuss,
only merges adjacent elements, $\QQ_i$ and~$\QQ_{i+1}$, on the stack.
This is necessary for the sort to be stable: If $i<i^\pprime<i^\prime$
and two non-adjacent runs
$\QQ_i$ and $\QQ_{i^\prime}$ were merged, then we would not know how to order
members occurring in both $\QQ_{i^\pprime}$ and $\QQ_i\cup \QQ_{i^\prime}$.
The patience sort, melsort, and split sort can all readily be modified
to be stable,
and our results on merge costs can be applied to them.

Our merge strategies
do not apply to non-stable merging,
but Barbay and Navarro~\cite{BarbayNavarro:AdaptiveSorting} have given an
optimal method---based on Huffmann codes---of merging for non-stable
merge sorts in which merged runs do not need to be adjacent.

The known worst-case upper and lower bounds
on stable natural merge sorts are listed in
Table~\ref{tab:allbounds}. The table expresses bounds in
the strongest forms known. Since $m \le n$, it is
generally preferable
to have upper bounds in
terms of $n \log m$, and lower bounds
in terms of $n \log n$.

\begin{table}[ht]
\begin{center}
\begin{tabular}{ccc}
Algorithm & Upper bound & Lower bound \\[1ex]
\hline
Timsort\footnotemark
    & $\left\{\begin{array}{c}
    O(n \log n) ~\hbox{\cite{AJNP:Timsort,ANP:MergeStrategies}} \\
    O(n \log m) ~\hbox{\cite{AJNP:Timsort}}
    \end{array}\right\}$
    & $1.5 \cdot n \log n$ [Theorem~\ref{thm:timsortLower}] \\[1ex]
$\alpha$-stack sort
    & $O(n \log n)$ \hbox{\cite{ANP:MergeStrategies}}
    & $\left\{ \begin{array}{ll}
               c_\alpha \cdot n \log n & \hbox{[Theorem~\ref{thm:alphaLower}]} \\
               \omega (n \log m) & \hbox{[Theorem~\ref{thm:alphaLowerLogm}]}
               \end{array}
       \right.$ \\[2ex]
Shivers sort
    & $\left\{ \begin{array}{c}
          \hbox{$n \log n$ \cite{Shivers:mergesort}} \\
          \hbox{See also Theorem~\ref{thm:ShiversUpperBound}}
      \end{array} \right\}$
    & \hbox{$\omega (n \log m)$ [Theorem~\ref{thm:ShiversLowerBound}]} \\[2ex]
$2$-merge sort
    & $c_2 \cdot n \log m$ [Theorem \ref{thm:twomergeUpper}]         
    & $c_2 \cdot n \log n$ [Theorem \ref{thm:twomergeLower}] \\[1ex] 
$\alpha$-merge sort
    & $c_\alpha \cdot n \log m$ [Theorem \ref{thm:alphamergeUpper}]  
    & $c_\alpha \cdot n \log n$ [Theorem \ref{thm:alphamergeLower}]  
\end{tabular}
\end{center}
\caption{Upper and lower bounds on the merge cost of various algorithms.
For more precise statements, see the theorems. The results hold for
$\varphi < \alpha \le 2$; for these values,
$c_\alpha$ is defined by equation~(\ref{eq:cAlphaDef}) and satisfies
$1.042<c_\alpha<1.089$.
In particular, $c_2 = 3/\log(27/4) \approx 1.08897$.
All bounds are asymptotic; that is, they are correct up to
a multiplicative factor of $1 \pm o(1)$.  For this reason,
the upper and lower bounds in the last two lines of the
the table are not exactly matching.
The table lists $2$-merge sort and $\alpha$-merge sort on
separate lines since $2$-merge sort is slightly simpler
than $\alpha$-merge sort.  In addition, our proof of the upper bound
for $2$-merge sort is substantially
simpler than our proof for $\alpha$-merge sort.}
\label{tab:allbounds}
\end{table}

\footnotetext{When the first
draft of the present paper was circulated, the question of an $O(n \log m)$
upper bound for Timsort was still open; this was subsequently resolved by~\cite{AJNP:Timsort}.}

The main results of the paper are those listed in the final two lines of
Table~\ref{tab:allbounds}. Theorem~\ref{thm:twomergeUpper} proves that the merge
cost of $2$-merge sort is at most $(d_2 + c_2\log m)\cdot n$, where $d_2\approx
1.911$ and $c_2\approx 1.089$: these are very tight bounds, and the value
for~$c_2$ is optimal by Theorem~\ref{thm:twomergeLower}. It is also
substantially better than the worst-case merge cost for Timsort proved in
Theorem~\ref{thm:timsortLower}. Similarly for $\varphi < \alpha < 2$,
Theorem~\ref{thm:twomergeUpper} proves an upper bound of $(d_\alpha+c_\alpha\log
m)\cdot n$. The values for $c_\alpha$ are optimal by
Theorem~\ref{thm:alphamergeLower}; however, our values for $d_\alpha$ have
unbounded limit $\lim_{\alpha \rightarrow \varphi^+}d_\alpha$ and we conjecture
this is not optimal.

We only analyze $\alpha$-merge sorts with $\alpha > \varphi$.
For $\varphi < \alpha \le 2$, the $\alpha$-merge sorts improve on Timsort,
by virtue of having better run time bounds and by being slightly easier
to implement. In addition, they perform better in the experiments reported
in Section~\ref{sec:experiments}.  It is an open
problem to extend our algorithms to the case of $\alpha < \varphi$; we expect
this will require $k$-aware algorithms with $k > 3$.

Table~\ref{tab:allbounds} does not include two new stable merge sorts
that were developed by Munro-Wild~\cite{MunroWild:MergeSort}
and Jug\'e~\cite{Juge:AdaptiveShivers} subsequently to the circulation
of the first draft of the present paper. Jug\'e's ``adaptive Shivers'' sort
is 3-aware. Munro and Wild's ``powersort'' algorithm is 3-aware except that it requires
knowing the value of~$n$ ahead of time. Both of these sorting alorithms have
asymptotic lower and upper bounds on running time of $n \log m$.

The outline of the paper is as follows. Section~\ref{sec:timsort} describes
Timsort, and proves the lower bound on its merge cost. Section~\ref{sec:ANP}
discusses the $\alpha$-stack sort algorithms, and gives lower bounds on their
merge cost. Section~\ref{sec:shivers} describes the Shivers sort, and gives a
simplified proof of the $n\log n$ upper bound of~\cite{Shivers:mergesort}.
Section~\ref{sec:alphaMerge} is the core of the paper and describes the new
$2$-merge sort and
$\alpha$-merge sort. We first a simple $O(n \log m)$ upper bound on
their merge cost. After that, we prove the lower bounds on their merge costs, and
finally prove the corresponding sharp upper bounds. Section~\ref{sec:experiments}
gives some experimental results on two kinds of randomly generated data. It includes
experimental comparisons also with the recent adaptive Shivers sort of Jug\'e and
powersort of Munro and Wild.  All
these sections can be read independently of each other. The paper concludes with
discussion of open problems.

We thank the referees of~\cite{BussKnop:StableMergeSortSODA}
for useful comments and suggestions, and Vincent Jug\'e
for sharing information on his unpublished work.

\section{Timsort lower bound}\label{sec:timsort}

Algorithm~\ref{alg:timsort} is the Timsort algorithm
as defined by~\cite{GRBBH:TimSortBug} improving
on~\cite{Peters:TimSortPost}. Recall
that $W,X,Y,Z$ are the top four elements
on the stack~$\mathcal \QQ$. A command ``Merge $Y$ and $Z$''
creates a single run which replaces both $Y$ and~$Z$ in the
stack; at the same time, the current third member on the
stack, $X$, becomes the new second member on the stack
and is now designated~$Y$.  Similarly, the current $W$
becomes the new $X$, etc.  Likewise, the command
``Merge $X$ and~$Y$'' merges the second and third
elements at the top of~$\mathcal \QQ$; those two elements
are removed from~$\mathcal \QQ$ and replaced by the result
of the merge.

\begin{algorithm}
\begin{algorithmic}[1]
\Procedure{$\text{TimSort}$}{$S$, $n$}
  \State $\mathcal{R} \gets \text{run decomposition of }S$
  \State $\mathcal \QQ \gets \emptyset$
  \While{$\mathcal{R} \neq \emptyset$}
    \State Remove the next run $R$ from $\mathcal{R}$ and push it onto~$\mathcal \QQ$
    \Loop
      \If{$|X| < |Z|$}              \label{algline:timXZ}
        \State{Merge $X$ and $Y$}
      \ElsIf{$|X| \le |Y| + |Z|$}   \label{algline:timXYZ}
        \State{Merge $Y$ and $Z$}
      \ElsIf{$|W| \le |X| + |Y|$}   \label{algline:timWXY}
        \State{Merge $Y$ and $Z$}
      \ElsIf{$|Y| \le |Z|$}         \label{algline:timYZ}
        \State{Merge $Y$ and $Z$}
      \Else
        \State Break out of the loop
      \EndIf
    \EndLoop
  \EndWhile
  \While{$|\mathcal \QQ| \ge 1$}
    \State{Merge $Y$ and $Z$}
  \EndWhile
\EndProcedure
\end{algorithmic}
\caption{The Timsort algorithm. $W,X,Y,Z$ denote the top four elements
of the stack~$\mathcal \QQ$. A test involving a stack member that does not
exist evaluates as ``False''. For example, $|X| < |Z|$ evaluates as false
when $|\mathcal \QQ| < 3$ and $X$ does not exist.}
\label{alg:timsort}
\end{algorithm}

Timsort was designed so that the stack has size $O(\log n)$,
and the total running time is $O(n\log n)$.
These bounds were first proved by~\cite{ANP:MergeStrategies};
simplified proofs were
given by~\cite{AJNP:Timsort} who also strengthed the
upper bound to $O(n \log m)$.

\begin{theorem}[\cite{AJNP:Timsort,ANP:MergeStrategies}]
    \label{thm:timsortUpper}
    The merge cost of Timsort is $O(n \log n)$, and even $O(n \log m)$.
\end{theorem}

\noindent
The proof in
Auger et al.~\cite{ANP:MergeStrategies} did not compute the constant implicit
in their proof of the upper bound of Theorem~\ref{thm:timsortUpper}; but it is
approximately equal to $3 / \log \varphi \approx 4.321$. The proofs
in~\cite{AJNP:Timsort} also do not quantify the constants in the big-O
notation, but they are comparable or slightly larger.
We prove a
corresponding lower bound.

\begin{theorem}\label{thm:timsortLower}
  The worst-case merge cost of the Timsort algorithm on inputs of length~$n$
  which decompose into $m$ original runs is $\ge (1.5 - o(1))\cdot n \log n$.
  Hence it is also $\ge (1.5 - o(1))\cdot n \log m$.
\end{theorem}

In other words, for any $c < 1.5$, there are
inputs to Timsort with arbitrarily large values for
$n$ (and~$m$) so that Timsort has merge cost
$> c \cdot n \log n$.
We conjecture that
Theorem~\ref{thm:timsortUpper} is nearly optimal:
\begin{conjecture}\label{conj:timsortUpper}
  The merge cost of Timsort is bounded by $(1.5 +  o(1))\cdot n \log m$.
\end{conjecture}
Vincent Jug\'e [personal communication, November 2018] has recently
succeeded in proving this conjecture.

\begin{proof}[Proof of Theorem~\ref{thm:timsortLower}]
  We must define inputs that cause Timsort to take time close to
  $1.5 n \log n$.  As always, $n\ge 1$ is the length of the input~$S$ to be
  sorted. We define $\Rtim (n)$ to be a sequence of run {\em lengths} so that
  $\Rtim(n)$ equals $\langle n_1, n_2 \ldots, n_m\rangle$ where each $n_i > 0$
  and  $\sum_{i = 1}^m n_i = n$.  Furthermore, we will have $m\le n \le 3m$, so
  that  $\log n = \log m + O(1)$. The notation $\Rtim$ is reminiscent of
  $\mathcal R$, but $\mathcal R$ is a sequence of runs whereas $\Rtim$ is a
  sequence of run lengths. Since the merge cost of Timsort depends only on the
  lengths of the runs, it is more convenient to work directly with the sequence
  of run lengths.

  The sequence $\Rtim(n)$, for $1\le n$, is defined as follows.\footnote{%
  For purposes of this proof, we allow run lengths to equal~1. Strictly
  speaking, this cannot occur since all original runs will have length at
  least~2. This is unimportant for the proof however, as the
  run lengths~$\Rtim(n)$ could be doubled and the asymptotic analysis needed
  for the proof would be essentially unchanged.}
  First, for $n\le 3$, $\Rtim(n)$ is the sequence~$\langle n \rangle$,
  i.e., representing a single run of length~$n$.
  Let $n^\prime = \lfloor n/2\rfloor$.
  For even $n \ge 4$, we have $n = 2n^\prime$ and
  define $\Rtim(n)$ to be the concatenation of
  $\Rtim(n^\prime)$, $\Rtim(n^\prime - 1)$ and $\langle 1 \rangle$.
  For odd $n\ge 4$, we have $n = 2n^\prime +  1$
  and define $\Rtim(n)$ to be the concatenation of $\Rtim(n^\prime)$,
  $\Rtim(n^\prime - 1)$ and~$\langle 2\rangle$.

  We claim that for $n\ge 4$, Timsort operates with run lengths $\Rtim(n)$ as
  follows: The first phase processes the runs from $\Rtim(n^\prime)$
  and merges them into a single run of length $n^\prime$
  which is the only element of the stack~$\mathcal \QQ$.
  The second phase processes the runs from $\Rtim(n^\prime - 1)$
  and merges them also into a single run of length $n^\prime - 1$;
  at this point the stack contains two runs, of lengths $n^\prime$
  and $n^\prime - 1$. Since $n^\prime - 1 < n^\prime$, no further merge
  occurs immediately. Instead, the final run is loaded onto the stack: it
  has length $n^\pprime$ equal to
  either 1 or~2. Now $n^\prime \le n^\prime - 1 + n^\pprime$ and
  the test $|X| \le |Y| + |Z|$ on line~\ref{algline:timXYZ}
  of Algorithm~\ref{alg:timsort} is triggered, so Timsort merges
  the top two elements of the stack, and then the
  test $|Y|\le|Z|$ causes the merge of the final two elements
  of the stack.

  This claim follows from
  Claim~\ref{clm:timsortLower}. We say that the stack~$\mathcal \QQ$ is
  {\em stable} if none of the tests on lines
  \ref{algline:timXZ}, \ref{algline:timXYZ}, \ref{algline:timWXY}, \ref{algline:timYZ},
  of Algorithm~\ref{alg:timsort} hold.

  \begin{claim}\label{clm:timsortLower}
  Suppose that $\Rtim(n)$ is the initial subsequence of
  a sequence~$\mathcal R^\prime$ of run lengths,
  and that Timsort is initially started with
  run lengths~$\mathcal R^\prime$ either
  (a)~with the stack~$\mathcal \QQ$ empty or
  (b)~with the top element of~$\mathcal \QQ$ a run of length $n_0>n$
  and the second element of~$\mathcal \QQ$ (if it exists)
  a run of length $n_1 > n_0 +  n$.
  Then Timsort will start by processing exactly the runs whose lengths
  are those of $\Rtim(n)$, merging them into a single run
  which becomes the new top element of~$\mathcal \QQ$.
  Timsort will do this without performing any merge of runs
  that were initially in~$\mathcal \QQ$ and without
  (yet) processing any of the remaining runs in~$\mathcal R^\prime$.
  \end{claim}

  Claim~\ref{clm:timsortLower} is proved by induction on~$n$. The base case, where
  $n\le 3$, is trivial since with $\mathcal \QQ$ stable, Timsort immediately reads
  in the first run from~$\mathcal R^\prime$. The case of $n\ge 4$ uses the
  induction hypothesis twice, since $\Rtim(n)$ starts off with $\Rtim(n^\prime)$
  followed by $\Rtim(n^\prime  -  1)$. The induction hypothesis applied to
  $n^\prime$ implies that the runs of $\Rtim(n^\prime)$ are first processed and
  merged to become the top element of~$\mathcal \QQ$. The stack elements $X,Y,Z$
  have lengths $n_1,n_0,n^\prime$ (if they exist), so the stack is now stable. Now
  the induction hypothesis for $n^\prime - 1$ applies, so Timsort next loads and
  merges the runs of $\Rtim(n^\prime  -  1)$. Now the top stack elements $W,X,Y,Z$ have
  lengths $n_1,n_0,n^\prime,n^\prime  -  1$ and $\mathcal \QQ$~is again stable.
  Finally, the single run of length~$n^\pprime$ is loaded onto the stack. This
  triggers the test $|X|\le |Y| + |Z|$, so the top two elements are merged. Then the
  test $|Y| \le |Z|$ is triggered, so the top two elements are again merged. Now
  the top elements of the stack (those which exist) are runs of length
  $n_1,n_0,n$, and Claim~\ref{clm:timsortLower} is proved.

  Let $c(n)$ be the merge cost of the Timsort algorithm on the sequence~$\Rtim(n)$
  of run lengths. The two merges described at the end of the proof of
  Claim~\ref{clm:timsortLower} have merge cost $(n^\prime - 1) + n^\pprime$ plus
  $n^\prime + (n^\prime - 1) + n^\pprime = n$. Therefore, for $n > 3$, $c(n)$
  satisfies
  \begin{equation} \label{eq:ctimrecur}
      c(n)
      ~=~
      \begin{cases}
          c(n^\prime) + c(n^\prime - 1) + \frac{3}{2} n ~
          & \hbox{if $n$ is even} \\[0.2ex]
          c(n^\prime) + c(n^\prime - 1) + \frac{3}{2} n + \frac 12 ~
          & \hbox{if $n$ is odd. } \\
      \end{cases}
  \end{equation}
  Also, $c(1) = c(2) = c(3) = 0$ since no merges are needed.
  Equation~(\ref{eq:ctimrecur}) can be summarized as
  \[
      c(n) ~=~
          c(\lfloor n / 2 \rfloor) + c(\lfloor n/2 \rfloor  -  1)
          + {\textstyle \frac 32} n + {\textstyle \frac 12} (n \bmod 2).
  \]
  The function $n \mapsto \frac 32 n + \frac 12 (n \bmod 2)$ is strictly
  increasing. So, by induction, $c(n)$ is strictly increasing for $n \ge 3$.
  Hence $c(\lfloor n/2 \rfloor) > c(\lfloor n/2 \rfloor  -  1)$, and
  thus $c(n) \ge 2 c(\lfloor n/2 \rfloor  -  1) + {\textstyle \frac 32}n$ for
  all $n > 3$.

  For $x\in \mathbb R$, define $b(x) = c(\lfloor x - 3 \rfloor)$. Since $c(n)$ is
  nondecreasing, so is $b(x)$. Then
  \begin{eqnarray*}
      b(x) &=& c(\lfloor x  -  3 \rfloor)
        ~\ge~ 2 c(\lfloor (x  -  3)/2 \rfloor - 1 ) + \frac 32 (x  -  3) \\
        &\ge& 2 c(\lfloor x / 2 \rfloor - 3 ) + \frac 32 (x  -  3)
        ~=~ 2 b(x / 2) + \frac 32 (x  -  3).
  \end{eqnarray*}
  \begin{claim} \label{clm:timsortLower2}
      For all $x\ge 3$,
      $b(x) \ge \frac 32\cdot [x(\lfloor \log x \rfloor-2) - x + 3 ]$.
  \end{claim}
  We prove the claim by induction, namely by induction on~$n$ that it holds for
  all $x<n$. The base case is when $3 \le x < 8$ and is trivial since the lower
  bound is negative and $b(x)\ge0$. For the induction step, the claim is known to
  hold for $x/2$. Then, since $\log (x / 2) = (\log x) - 1$,
  \begin{eqnarray*}
    b(x) &\ge& 2 \cdot b(x/2) + \frac 32 (x - 3) \\
         &\ge& 2 \cdot \Bigl( \frac 32 \cdot \bigl[ (x/2)(\lfloor \log x \rfloor - 3)
                               - x/2 + 3 \bigr] \Bigr)
               + \frac 32 ( x - 3 ) \\
         &=& \frac 32 \cdot [ x(\lfloor \log x \rfloor-2) - x + 3 ]
  \end{eqnarray*}
  proving the claim.

  Claim~\ref{clm:timsortLower2} implies that $c(n) = b(n +  3) \ge
  (\frac 32 - o(1))\cdot n \log n$.
  This proves Theorem~\ref{thm:timsortLower}.
\end{proof}

\section{The \texorpdfstring{$\alpha$}{a}-stack sort}\label{sec:ANP}

Augur-Nicaud-Pivoteau~\cite{ANP:MergeStrategies} introduced the $\alpha$-stack
sort as a $2$-aware stable merge sort; it was inspired by Timsort and designed
to be simpler to implement and to have a simpler analysis. (The algorithm~(e2)
of~\cite{ZML:SortRace} is the same as $\alpha$-stack sort with $\alpha=2$.)
Let $\alpha > 1$ be a constant.  The $\alpha$-stack sort is shown in
Algorithm~\ref{alg:alphaStackSort}. It makes less effort than Timsort to
optimize the order of merges: up until the run decomposition is exhausted, its
only merge rule is that $Y$ and $Z$ are merged whenever $|Y| \le \alpha |Z|$. An
$O(n \log n)$ upper bound on its runtime is given by the next theorem.

\begin{algorithm}
    \begin{algorithmic}[1]
        \Procedure{$\alpha\text{-stack}$}{$S$, $n$}
          \State $\mathcal{R} \gets \text{run decomposition of }S$
          \State $\mathcal \QQ \gets \emptyset$
          \While{$\mathcal{R} \neq \emptyset$}
            \State Remove the next run $R$ from $\mathcal{R}$ and push it onto~$\mathcal \QQ$
            \While{$|Y| \le \alpha |Z|$} \label{algline:astackYZ}
              \State{Merge $Y$ and $Z$}  \label{algline:astackMergeYZ1}
            \EndWhile                    \label{algline:astackinnerloopB}
          \EndWhile
          \While{$|\mathcal \QQ| \ge 1$}  \label{algline:astackfinalloopA}
            \State{Merge $Y$ and $Z$}
          \EndWhile                      \label{algline:astackfinalloopB}
        \EndProcedure
    \end{algorithmic}
    \caption{The $\alpha$-stack sort. $\alpha$ is a constant $>1$.}
    \label{alg:alphaStackSort}
\end{algorithm}

\begin{theorem}[\cite{ANP:MergeStrategies}]\label{thm:alphaUpper}
    Fix $\alpha > 1$. The merge cost for $\alpha$-stack sort is
    $O(n \log n)$.
\end{theorem}
\cite{ANP:MergeStrategies}~did not explicitly mention the
constant implicit in this upper bound, but their proof establishes
a constant equal to approximately $(1 +  \alpha)/\log \alpha$. For instance,
for $\alpha=2$, the merge cost is bounded by
$(3 +  o(1))n \log n$. The constant is minimized at
$\alpha\approx 3.591$, where is it approximately 2.489.

\begin{theorem}\label{thm:alphaLower}
    Let $1 < \alpha$.
    The worst-case merge cost of $\alpha$-stack sort
    on inputs of length~$n$
    is $\ge (c_\alpha - o(1))\cdot n \log n$,
    where
    $c_\alpha$ equals
    $\frac{\alpha + 1}{(\alpha + 1) \log(\alpha + 1) - \alpha \log(\alpha)}$.
\end{theorem}

The proof of Theorem~\ref{thm:alphaLower} is postponed
until Theorem~\ref{thm:alphamergeLower} proves a
stronger lower bound for
$\alpha$-merge sorts; the same construction works to
prove both theorems. The value $c_\alpha$ is quite
small, e.g., $c_2 \approx 1.089$; this is
is discussed more in Section~\ref{sec:alphaMerge}.

The lower bound of Theorem~\ref{thm:alphaLower} is not very strong since the
constant is close to~1.  In fact, since a binary tree of merges gives a merge
cost of $n \lceil\log m \rceil$, it is more relevant to give upper bounds in terms of
$O(n \log m)$ instead of $O(n \log n)$.  The next theorem shows that
$\alpha$-stack sort can be very far from optimal in this respect.

\begin{theorem}\label{thm:alphaLowerLogm}
  Let $1 < \alpha$. The worst-case merge cost of $\alpha$-stack sort on
  inputs of length~$n$ which decompose into $m$ original runs is
  $\omega(n \log m)$.
\end{theorem}
\noindent
In other words, for any $c > 0$, there are
inputs with arbitrarily large values for
$n$ and~$m$ so that $\alpha$-stack sort has merge cost
$> c \cdot n \log m$.

\begin{proof}
  Let $s$ be the least integer such that $2^s\ge \alpha$. Let $\Rastack(m)$ be
  the sequence of run lengths
  \[
    \langle \, 2^{(m - 1)\cdot s}  -  1, \,
             2^{(m - 2)\cdot s}  -  1, \,
             \ldots, 2^{3s} - 1, \, 2^{2s} - 1, \, 2^s - 1, \,
             2^{m\cdot s}
    \rangle.
  \]
  $\Rastack(m)$ describes $m$ runs whose lengths sum to
  $n = 2^{m\cdot s} + \sum_{i = 1}^{m - 1} (2^{i \cdot s} - 1)$,
  so $ 2^{m\cdot s} < n < 2^{m\cdot s+1}$.
  Since $2^s\ge \alpha$, the test $|Y| \le \alpha |Z|$ on
  line~\ref{algline:astackYZ} of
  Algorithm~\ref{alg:alphaStackSort} is triggered
  only when the run of length $2^{ms}$ is loaded onto
  the stack~$\mathcal \QQ$; once this happens the runs are all
  merged in order from right-to-left.
  The total cost of the merges is
  $(m - 1)\cdot 2^{ms} + \sum_{i=1}^{m-1} i\cdot (2^{i\cdot s}  -  1)$
  which is certainly greater than $(m - 1)\cdot 2^{ms}$.  Indeed,
  that comes from the fact that the final run in~$\Rastack(m)$ is
  involved in $m-1$ merges.  Since $n<2\cdot 2^{m\cdot s}$,
  the total merge cost is greater than $\frac n2(m - 1)$, which
  is $\omega(n\log m)$.
\end{proof}

\section{The Shivers merge sort}\label{sec:shivers}

The 2-aware Shivers sort~\cite{Shivers:mergesort}, shown in
Algorithm~\ref{alg:shiversSort},
is similar to $2$-stack sort, but
with a modification that makes a surprising improvement
in the bounds on its merge cost. Although never published,
this algorithm was presented in 1999.

\begin{algorithm}
\begin{algorithmic}[1]
\Procedure{ssort}{$S$, $n$}
  \State $\mathcal{R} \gets \text{run decomposition of }S$
  \State $\mathcal \QQ \gets \emptyset$
  \While{$\mathcal{R} \neq \emptyset$}     \label{algline:shiversLoopA}
    \State Remove the next run $R$ from $\mathcal{R}$ and push it onto~$\mathcal \QQ$
    \While{$2^{\lfloor{\log |Y|}\rfloor} \le |Z|$}  \label{algline:shiversYZ}
      \State{Merge $Y$ and $Z$}
    \EndWhile                               \label{algline:ShiversInnerLoopB}
  \EndWhile                                \label{algline:shiversLoopB}
  \While{$|\mathcal \QQ| \ge 1$} \label{algline:shiversEndA}
    \State{Merge $Y$ and $Z$}
  \EndWhile                     \label{algline:shiversEndB}
\EndProcedure
\end{algorithmic}
\caption{The Shivers sort.}
\label{alg:shiversSort}
\end{algorithm}

The only difference between the Shivers sort
and $2$-stack sort is the test used to decide
when to merge. Namely, line~\ref{algline:shiversYZ}
tests $2^{\lfloor \log |Y|\rfloor} \le |Z|$ instead
of $|Y| \le 2\cdot|Z|$. Since $2^{\lfloor \log |Y|\rfloor}$
is $|Y|$ rounded down to the nearest power of two,
this is somewhat like an $\alpha$-sort with
$\alpha$ varying dynamically in the range $[1,2)$.

The Shivers sort has the same
undesirable lower
bound as $2$-stack sort in terms of $\omega(n \log m)$:

\begin{theorem}\label{thm:ShiversLowerBound}
  The worst-case merge cost of the Shivers sort
  on inputs of length~$n$ which decompose into $m$
  original runs is $\omega(n \log m)$.
\end{theorem}

\begin{proof}
  This is identical to the proof of Theorem~\ref{thm:alphaLowerLogm}.
  We now let $\Rshiver(m)$ be the sequence of run lengths
  \[
    \langle \, 2^{m - 1}  -  1, \, 2^{m - 2}  -  1,
      \, \dots, 7, \, 3, \, 1, \, 2^m \rangle,
  \]
  and argue as before.
\end{proof}

\begin{theorem}[\cite{Shivers:mergesort}]\label{thm:ShiversUpperBound}
  The merge cost of Shivers sort is $(1  +   o(1)) n \log n$.
\end{theorem}

We present a proof which is simpler than that of~\cite{Shivers:mergesort}.
The proof of Theorem~\ref{thm:ShiversUpperBound}
assumes that at a given point in time, the stack~$\mathcal \QQ$ has $\ell$
elements $\QQ_1, \ldots, \QQ_\ell$, and uses $w_{\QQ_i}$ to denote the merge cost of~$\QQ_i$.
We continue to use the convention that $W,X,Y,Z$ denote $\QQ_{\ell - 3},
\QQ_{\ell - 2}, \QQ_{\ell - 1}, \QQ_\ell$ if they exist.

\begin{proof}
  Define $k_{\QQ_i}$ to equal $\lfloor \log |\QQ_i| \rfloor$.
  Obviously, $|\QQ_i| \ge 2^{k_{\QQ_i}}$. The
  test on line~\ref{algline:shiversYZ} works to
  maintain the invariant that each $|\QQ_{i+1}| < 2^{k_{\QQ_i}}$
  or equivalently $k_{i+1}<k_i$.
  Thus, for $i<\ell - 1$, we always
  have $|\QQ_{i+1}| < 2^{k_{\QQ_i}}$ and $k_{i+1}<k_i$.
  This condition can be momentarily violated for $i=\ell - 1$,
  i.e.\ if $|Z|\ge 2^{k_Y}$ and $k_Y\le k_Z$,
  but then the Shivers sort immediately merges $Y$ and~$Z$.

  As a side remark, since each $k_{i+1}<k_i$ for $i\le \ell - 1$,
  since $2^{k_1}\le|\QQ_1|\le n$,
  and since $\QQ_{\ell - 1}\ge 1 = 2^0$, the stack height $\ell$
  is $\le 2+\log n$. (In fact, a better analysis shows
  it is $\le 1+\log n$.)

  \begin{claim}\label{clm:ShiversUpperBound}
    Throughout the execution of the main loop (lines
    {\rm \ref{algline:shiversLoopA}-\ref{algline:shiversLoopB}}), the Shivers sort
    satisfies
    \begin{itemize}
      \setlength{\parsep}{0pt}
      \setlength{\itemsep}{0pt}
      \item[\rm a.] $w_{\QQ_i} \le k_{\QQ_i}\cdot |\QQ_i|$, for all $i\le\ell$,
      \item[\rm b.] $w_Z \le k_Y \cdot |Z|$ i.e.,~$w_{\QQ_\ell} \le
        k_{\QQ_{\ell - 1}} \cdot |\QQ_\ell|$, if $\ell > 1$.
    \end{itemize}
  \end{claim}
  \noindent
  When $i = \ell$, a.\ says $w_Z \le k_Z |Z|$.
  Since $k_Z$ can be less than or greater than~$k_Y$, this
  neither implies, nor is implied by,~b.

  The lemma is proved by induction on the number of updates to the
  stack~$\mathcal \QQ$ during the loop. Initially $\mathcal \QQ$ is empty,
  and a.\ and~b.\ hold trivially.
  There are two induction cases to consider. The first case is when
  an original run is
  pushed onto~$\mathcal \QQ$. Since this run, namely~$Z$, has never
  been merged, its weight is $w_Z=0$. So b.\ certainly holds.
  For the same reason and using the induction hypothesis, a.\ holds.
  The second case is when $2^{k_Y} \le |Z|$, so $k_Y\le k_Z$, and
  $Y$ and~$Z$ are merged; here $\ell$ will decrease by~1.
  The merge cost~$w_{\YZ}$ of the combination of $Y$ and $Z$ equals
  $|Y| + |Z|+w_Y+w_Z$, so
  we must establish two things:
  \begin{itemize}
      \item[\rm a${}^\prime$.] $|Y| + |Z|+w_Y+w_Z ~\le~ k_{\YZ} \cdot (|Y| +  |Z|)$, where
          $k_{\YZ}= \lfloor \log(|Y| +  |Z|)\rfloor$.
      \item[\rm b${}^\prime$.] $|Y| + |Z|+w_Y+w_Z ~\le~ k_X \cdot (|Y| +  |Z|)$, if $\ell>2$.
  \end{itemize}
  By induction hypotheses $w_Y \le k_Y |Y|$ and
  $w_Z \le k_Y|Z|$.  Thus the lefthand sides of a${}^\prime$.\ and~b${}^\prime$.\
  are $\le (k_Y+1)\cdot(|Y| + |Z|)$.
  As already discussed, $k_Y< k_X$, therefore condition~b.\ implies that
  b${}^\prime$.\ holds. And since
  $2^{k_Y} \le |Z|$, $k_Y < k_{\YZ}$, so condition~a.\ implies that
  a${}^\prime$.\ also holds.  This
  completes the proof of Claim~\ref{clm:ShiversUpperBound}.

  Claim~\ref{clm:ShiversUpperBound} implies that
  the total merge cost incurred at the end of the main loop incurred is
  $\le \sum_i w_{\QQ_i}  \le \sum_i k_i |\QQ_i|$. Since $\sum_i \QQ_i = n$
  and each $k_i \le \log n$, the
  total merge cost is $\le n \log n$.

  We now upper bound the total merge cost incurred during the
  final loop on lines \ref{algline:shiversEndA}-\ref{algline:shiversEndB}.
  When first reaching line~\ref{algline:shiversEndA},
  we have $k_{i+1}<k_i$ for all $i\le \ell - 1$ hence
  $k_i<k_1+1-i$ and $|\QQ_i|<2^{k_1+2-i}$ for all $i\le \ell$. The final loop then
  performs $\ell - 1$ merges from
  right to left. Each $\QQ_i$ for $i < \ell$ participates
  in $i$ merge operations and $\QQ_\ell$ participates
  in $\ell - 1$ merges. The total merge cost of this is less than
  $\sum_{i=1}^\ell i\cdot|\QQ_i|$. Note that
  \[
      \sum_{i=1}^\ell i\cdot|\QQ_i|
      ~<~
      2^{k_1+2}\cdot \sum_{i=1}^\ell i \cdot 2^{-i}
      ~<~
      2^{k_1+2} \cdot 2 ~=~ 8\cdot 2^{k_1} ~\le~ 8n,
  \]
  where the last inequality follows by $2^{k_1} \le |\QQ_1| \le n$. Thus, the
  final loop incurs a merge cost $O(n)$, which is $o(n \log n)$.

  Therefore the total merge cost for the Shivers sort is bounded by $n \log n
  + o(n\log n)$.
\end{proof}

\section{The 2-merge and \texorpdfstring{$\alpha$}{a}-merge sorts}
\label{sec:alphaMerge}

This section introduces our new merge sorting algorithms, called the
``2-merge sort'' and the
``$\alpha$-merge sort'', where $\varphi < \alpha < 2$ is a fixed parameter.
These sorts are 3-aware, and this enables us to get algorithms with merge costs
$(c_\alpha \pm o(1)) n \log m$. The idea of the $\alpha$-merge sort is to
combine the construction of $2$-stack sort, with the idea from
Timsort of merging $X$ and~$Y$ instead of $Y$ and~$Z$ whenever $|X| < |Z|$.  But
unlike the Timsort algorithm shown in Algorithm~\ref{alg:timsort}, we are able
to use a 3-aware algorithm instead of a 4-aware algorithm. In addition, our
merging rules are simpler, and our provable upper bounds are tighter. Indeed,
our upper bounds for $\varphi < \alpha\le 2$ are of the form
$(c_\alpha  +   o(1))\cdot n \log m$ with $c_\alpha \le c_2 \approx 1.089$, but
Theorem~\ref{thm:timsortLower} proves a lower bound $1.5 \cdot n \log m$ for
Timsort.

Algorithms~\ref{alg:twoMergeSort} and~\ref{alg:alphaMergeSort} show
the $2$-merge sort and $\alpha$-merge
sort algorithms. Note that $2$-merge sort is almost, but not quite, the
specialization of $\alpha$-merge sort to the case $\alpha=2$.
The difference is that line~\ref{algline:2mergeXYZ} of
$2$-merge sort has a simpler {\bf while} test than the
corresponding line in the $\alpha$-merge sort algorithm.
As will be shown by the proof of Theorem~\ref{thm:twomergeUpper},
the fact that Algorithm~\ref{alg:twoMergeSort} uses this
simpler {\bf while} test makes no difference to which merge operations
are performed; in other words, it would be redundant to test the
condition $|X| < 2|Y|$.

The $2$-merge sort can also be compared to the
$\alpha$-stack sort shown in Algorithm~\ref{alg:alphaStackSort}.
The main difference is that the merge of $Y$ and~$Z$ on
line~\ref{algline:astackMergeYZ1}
of the $\alpha$-stack sort algorithm has been replaced
by the lines
lines~\ref{algline:2mergeXZ}-\ref{algline:2mergeXZend}
of the $2$-merge sort algorithm which conditionally merge~$Y$ with either $X$ or~$Z$.
For $2$-merge sort
(and $\alpha$-merge sort), the run~$Y$ is
never merged with~$Z$ if it could instead be merged with
a shorter~$X$.  The other, perhaps less crucial, difference is that the weak
inequality test on line~\ref{algline:astackYZ}
in the $\alpha$-stack sort algorithm has been replaced with a
strict inequality test on line~\ref{algline:astackYZ}
in the $\alpha$-merge sort algorithm. We have made this change since
it seems to make $2$-merge sort more efficient, for instance when
all original runs have the same length.

\begin{algorithm}
    \begin{algorithmic}[1]
        \Procedure{$2$-merge}{$S$, $n$}
          \State $\mathcal{R} \gets \text{run decomposition of }S$
          \State $\mathcal \QQ \gets \emptyset$
          \While{$\mathcal{R} \neq \emptyset$}
            \State Remove the next run $R$ from $\mathcal{R}$ and push it
            onto~$\mathcal \QQ$
            \While{$|Y| < 2 |Z|$}\label{algline:2mergeXYZ}
              \If{$|X| < |Z|$}                              \label{algline:2mergeXZ}
                \State{Merge $X$ and $Y$}
              \Else
                \State{Merge $Y$ and $Z$}
              \EndIf                                        \label{algline:2mergeXZend}
            \EndWhile                       
          \EndWhile
          \While{$|\mathcal \QQ| \ge 1$}     \label{algline:2mergefinalloopA}
            \State{Merge $Y$ and $Z$}
          \EndWhile                         \label{algline:2mergefinalloopB}
        \EndProcedure
    \end{algorithmic}
    \caption{The $2$-merge sort.}
    \label{alg:twoMergeSort}
\end{algorithm}

\begin{algorithm}
    \begin{algorithmic}[1]
        \Procedure{$\alpha\text{-merge}$}{$S$, $n$}
          \State $\mathcal{R} \gets \text{run decomposition of }S$
          \State $\mathcal \QQ \gets \emptyset$
          \While{$\mathcal{R} \neq \emptyset$}
            \State Remove the next run $R$ from $\mathcal{R}$ and push it
            onto~$\mathcal \QQ$
            \While{$|Y| < \alpha |Z|$ or $|X| < \alpha |Y|$}\label{algline:amergeXYZ}
              \If{$|X| < |Z|$}                              \label{algline:amergeXZ}
                \State{Merge $X$ and $Y$}
              \Else
                \State{Merge $Y$ and $Z$}
              \EndIf
            \EndWhile                       \label{algline:amergeinnerloopB}
          \EndWhile
          \While{$|\mathcal \QQ| \ge 1$}     \label{algline:amergefinalloopA}
            \State{Merge $Y$ and $Z$}
          \EndWhile                         \label{algline:amergefinalloopB}
        \EndProcedure
    \end{algorithmic}
    \caption{The $\alpha$-merge sort.  $\alpha$ is a constant such that $\varphi<\alpha<2$.}
    \label{alg:alphaMergeSort}
\end{algorithm}

We will concentrate mainly on the
cases for $\varphi < \alpha\le 2$ where $\varphi \approx 1.618$
is the golden ratio.  Values for $\alpha>2$ do not seem to give
useful merge sorts; our upper bound proof does not
work for $\alpha \le \varphi$.

\begin{definition}
    Let $\alpha\ge 1$, the constant $c_\alpha$ is defined by
    \begin{equation}\label{eq:cAlphaDef}
        c_\alpha ~=~
            \frac{\alpha + 1}{(\alpha  +   1) \log(\alpha  +   1) -
            \alpha \log(\alpha)}.
    \end{equation}
\end{definition}
For $\alpha = 2$, $c_2 = 3/\log (27/4) \approx 1.08897$.
For $\alpha = \varphi$, $c_\varphi \approx 1.042298$.
For $\alpha > 1$, $c_\alpha$ is strictly increasing as a function
of~$\alpha$. Thus, $1.042 < c_\alpha < 1.089$ when $\varphi < \alpha \le 2$.

Sections~\ref{sec:alphaMergeLower}-\ref{sec:alphaMergeUpperPf}
give nearly matching upper and lower bounds for the
worst-case running time of $2$-merge sort and $\alpha$-merge sort for
$\varphi < \alpha < 2$.  As discussed earlier
(see also Table~\ref{tab:allbounds}) the upper and lower bounds are
asymptotically equal to $c_2\cdot n \log m$. The tight upper bound proofs for
Theorems~\ref{thm:alphamergeUpper} and~\ref{thm:twomergeUpper} involve a lot of
technical details, so as a warmup, Section~\ref{sec:simpleUpperBounds} gives
a weaker upper bound with a much simpler proof.

\subsection{Simple upper bounds for 2-merge sort and
  \texorpdfstring{$\alpha$}{a}-merge sort}
\label{sec:simpleUpperBounds}

This section proves $O(n \log n)$ upper bounds on the merge costs for
the $2$-merge and $\alpha$-merge sorts. The proofs are in spirit of
techniques used by~\cite{AJNP:Timsort,Juge:AdaptiveShivers}. First,
Theorem~\ref{thm:mergeStackSize} establishes an upper bound on the
size of the stack.

\begin{theorem}\label{thm:mergeStackSize}
  Let $\varphi < \alpha < 2$. For $\alpha$-merge sort on inputs of length~$n$,
  the size $|\mathcal \QQ|$ of the stack is always
  $< 1 + \lfloor \log_\alpha n\rfloor =
    1 + \lfloor (\log n) / (\log \alpha)\rfloor$.
  For $2$-merge sort, the size of the stack is always
  $\le 1 + \lfloor \log n\rfloor$.
\end{theorem}
\begin{proof}
  The main loop of the $2$-merge and $\alpha$-merge sorts first pushes
  a new run onto the stack and then does zero or more merges.
  Since merges reduce the stack height,
  the maximum stack height is obtained when a
  run is pushed onto the stack. Such a newly pushed
  run is denoted either $Z$ or $\QQ_\ell$, where the
  stack height is~$\ell$.

  We claim that for both $2$-merge sort and $\alpha$-merge sort,
  when a new run is pushed onto the
  stack, $|\QQ_{i - 1}| \ge \alpha |\QQ_i|$ holds for
  all $i < \ell$ (taking $\alpha=2$ in the inequality for $2$-merge sort).

  For $\alpha$-merge sort, this claim follows from the following
  three observations:
  \begin{description}
    \setlength{\itemsep}{0pt}
    \setlength{\topsep}{0pt}
    \item[\rm (a)] Each merge acts on two of the top
    three elements of the stack, i.e., either $X$ and $Y$
    are merged or $Y$ and~$Z$ are merged. Each merge decreases
    the height of the stack by~1; this decreases the value
    of~$\ell$ by~1.
    \item[\rm (b)] The condition that $|\QQ_{i - 1}| \ge \alpha |\QQ_i|$ holds
      for all $i < \ell - 1$ is maintained after merging two of the top three
      stack members. This fact is essentially trivial: If $Y$ and~$Z$ are merged
      they become the new top element of the stack; and if $X$ and~$Y$ are
      merged they become the second element of the stack. But the top
      two elements are $\QQ_{\ell - 1}$ and $\QQ_\ell$, and thus
      do not affect the conditions $|\QQ_{i - 1}| \ge \alpha |\QQ_i|$
      for $i < \ell - 1$.
    \item[\rm (c)] Merging continues until both
      $|X| \ge \alpha |Y|$ and $|Y| \ge \alpha |Z|$ hold.
  \end{description}

  Proving the claim for $2$-merge sort requires an extra argument
  since line~\ref{algline:2mergeXYZ} of Algorithm~\ref{alg:twoMergeSort}
  lacks a test for $|X|<2|Y|$. Instead, we argue that as $2$-merge sort
  performs merges, we have $|\QQ_{i-1}| \ge 2 |\QQ_i|$ for all $i<\ell - 1$
  and either
  \begin{itemize}
  \item[\rm (i)] $|\QQ_{\ell - 2}| \ge 2 |\QQ_{\ell - 1}|$, or
  \item[\rm (ii)] $|\QQ_{\ell - 2}| > | Q_{\ell - 1}|$ and
                  $|\QQ_{\ell - 1}| < 2 |\QQ_\ell|$.
  \end{itemize}
  The condition $|\QQ_{i-1}| \ge 2 |\QQ_i|$ for $i<\ell - 1$ is an invariant just
  because each step either merges $X$ and~$Y$ or merges $Y$ and~$Z$ (just as
  argued in (b) above).  For conditions (i) and~(ii), consider separately
  whether $X$ and~$Y$ are merged, or $Y$ and~$Z$ are merged. In either case,
  $\ell$~is decreased by one by the merge. When $Y$ and~$Z$ are merged, we
  trivially have that condition~(i) holds after the merge (as argued in (b)
  above). When $|X| < |Z|$, and $X$ and~$Y$ are merged, we have
  $|\QQ_{\ell - 3}| \ge 2 |\QQ_{\ell - 2}| >
    |\QQ_{\ell - 2}| + |\QQ_{\ell - 1}|$,
  so the first part of condition~(ii) still holds after the merge.
  In addition, $|Z|>|X|$ means the same as $|\QQ_\ell| > |\QQ_{\ell - 2}|$,
  so $2 |\QQ_\ell| > |\QQ_{\ell - 2}| + |\QQ_{\ell - 1}|$ so the second part
  of condition~(ii) also holds after the merge.  The merging stops
  when $|\QQ_{\ell - 1}| \ge2 |\QQ_\ell|$. Merging
  cannot stop while condition~(ii) holds, so it must stop while
  condition~(i) holds. The claim for $2$-merge sort follows immediately.

  For both $2$-merge sort and $\alpha$-merge sort, the claim implies
  that when a new run~$\QQ_\ell$ is pushed onto the stack, the other
  $\ell - 1$ runs $\QQ_1, \ldots, \QQ_{\ell - 1}$ have run lengths totaling
  \[
    \sum_{i = 1}^{\ell - 1} |\QQ_i|
     ~\ge~ |\QQ_{\ell - 1}| \cdot \sum_{i=0}^{\ell - 2} \alpha^i
     ~=~ \frac{\alpha^{\ell - 1} - 1}{\alpha - 1} \cdot |\QQ_{\ell - 1}|
  \]
  Therefore, the total run length of all $\ell$ runs on the stack
  at least $\frac{\alpha^{\ell - 1} - 1}{\alpha - 1}+1$.
  This must be $\le n$. Therefore, since $\alpha\le 2$, we have
  $\alpha^{\ell - 1} \le n$. Thus $\ell \le 1 + \lfloor \log_\alpha n\rfloor$.
  The inequalities are strict if $\alpha < 2$ since then $\alpha-1<1$.
\end{proof}

\begin{theorem}\label{thm:alphaMergeUpperEasy}
  Let $1 < \alpha \le 2$. The merge cost of $\alpha$-merge sort on
  inputs of length~$n$ is $O(n \log n)$.
\end{theorem}
\begin{proof}
  We introduce a special real-valued counter~$C$
  which initially is equal to~$0$. Every
  time we load a run $\QQ_\ell$ to the stack $\QQ$, we add
  $(2+\alpha) \cdot \ell\cdot |\QQ_\ell|$ to $C$. Every time we merge
  two runs $\QQ_i$ and $\QQ_{i + 1}$, we subtract $|\QQ_i| + |\QQ_{i + 1}|$ from~$C$.

  We claim that the following invariant
  holds throughout the execution of $\alpha$-merge sort:
  \begin{equation} \label{eq:Cinvariant}
    C ~\ge~ \sum\limits_{i = 1}^\ell (2+\alpha) i |\QQ_i|.
  \end{equation}
  We also claim this holds with $\alpha=2$ for $2$-merge sort.
  This claim will suffice to prove Theorem~\ref{thm:alphaMergeUpperEasy}.
  First, by Theorem~\ref{thm:mergeStackSize}, the total of the increases in~$C$
  caused by adding new runs is $\le (2 + \alpha)(1+\penalty10000 \log_\alpha(n))$.
  Second, the total decreases in~$C$ equal the total merge cost.
  By~(\ref{eq:Cinvariant}), $C$~is always non-negative,
  whence Theorem~\ref{thm:alphaMergeUpperEasy} follows.

  The invariant~(\ref{eq:Cinvariant}) holds initially, and it clearly
  is maintained by loading a new run to the top of the stack.
  Let us consider what happens when we
  make a merge. We need to consider several cases below. Only
  the first case is needed for $2$-merge sort as $2$-merge sort does not
  check the condition $|X|<2|Y|$ to decide whether to perform a merge.
  Cases 2.\ and 3.\ below apply only to $\alpha$-merge sort.
\begin{enumerate}
\item If $|Y| < \alpha |Z|$ then we merge $Y$ with the smaller of
$Z$ and~$X$, and subtract at most $|Y| + |Z|$ from~$C$. The merge
reduces the stack height by~1, so considering just the elements
of~$Z$, the righthand side of the inequality~(\ref{eq:Cinvariant})
decreases by at least $(2+\alpha)|Z|$. This decrease is
$\ge(2+\penalty10000\alpha)|Z| \ge |Y| +\penalty10000 |Z|$
by the assumption that $|Y| < \alpha |Z|$.
Thus the condition~(\ref{eq:Cinvariant}) is preserved in this case.
\item If a merge is triggerd by $|X| < \alpha |Y|$, there are several cases to
consider. Consider first the case where $|X| < |Z|$ so
we merge $X$ and~$Y$. In this case, the merge cost is
$|X| + |Y|$, so $C$ decreases by $|X| + |Y|$.
On the other hand, the righthand side of the
inequality~(\ref{eq:Cinvariant}) decreases by
$(2+\alpha)(|Y| + |Z|)$ since the elements of the runs
$Y$ and~$Z$ both have the stack height below them decrease by~1.
Since $|Z| < |X|$, the condition~(\ref{eq:Cinvariant}) is maintained.
\item Now suppose a merge is triggered by $|X| < \alpha |Y|$ and $|X|\ge|Z|$. In
this case we merge $Y$ and~$Z$ to form a run, denoted $YZ$, of length
$|Y| + |Z|$, and $C$~decreases by $|Y| + |Z|$. Now this step might
not preserve the condition~(\ref{eq:Cinvariant}). However,
after merging $Y$ and~$Z$, there must immediately be
another merge due to the fact that $X\le \alpha(|Y| + |Z|)$.
This next merge will either (a)~merge $W$ with~$X$ if $|W| < |Y| + |Z|$
or (b)~merge $X$ with $YZ$ if $|W|\ge |Y| + |Z|$.

Consider the case (b) where $X$ is merged with $YZ$. In this
case, $C$~has decreased by a total of $|X| + 2|Y| + 2|Z|$
during the two merges. At the same time, the stack heights
below the members of $Y$ and~$Z$ decreased by 1 and~2, respectively.
Therefore, the righthand side of~(\ref{eq:Cinvariant})
decreased by $(2+\alpha)(|Y|+2|Z|)$. Since $|X|\le \alpha|Y|$,
this is clearly greater than the decrease $|X| + 2|Y| + 2|Z|$
in~$C$. Thus the inequality~(\ref{eq:Cinvariant}) is
maintained in this case.

Now consider the case~(a) where $W$ and $X$ are merged.
In this case, $C$~decreases by $|W| + |X| + |Y| + |Z|$
from the two merges. Since $|W| < |Y| + |Z|$, this
means that $C$ decreases by less than $|X|+2|Y|+2|Z|$.
On the other hand, the reductions
in stack height means that the righthand side
of~(\ref{eq:Cinvariant}) decreases by
$(2+\alpha)(|X| + |Y|+2|Z|)$. This again implies that
the inequality~(\ref{eq:Cinvariant}) is maintained.
\end{enumerate}
\end{proof}

\subsection{Lower bound for 2-merge sort and \texorpdfstring{$\alpha$}{a}-merge sort}
\label{sec:alphaMergeLower}

\begin{theorem}\label{thm:alphamergeLower}
  Fix $\alpha>1$. The worst-case merge cost of $\alpha$-merge sort
  is $\ge (c_\alpha - o(1)) n \log n$.
\end{theorem}
\noindent
The corresponding theorem for $\alpha=2$ is:
\begin{theorem}\label{thm:twomergeLower}
  The worst-case merge cost of $2$-merge sort is
  $\ge (c_2  -  o(1)) n \log n$, where $c_2 = 3 / \log (27 / 4) \approx
  1.08897$.
\end{theorem}

The proof of Theorem~\ref{thm:alphamergeLower} also establishes
Theorem~\ref{thm:alphaLower}, as the same lower bound construction
works for both $\alpha$-stack sort and $\alpha$-merge sort.
The only difference is
that part~d.\ of Claim~\ref{clm:Ramerge} is used instead
of part~c.  In addition, the proof of Theorem~\ref{thm:alphamergeLower}
also establishes Theorem~\ref{thm:twomergeLower}; indeed, exactly the
same proof applies verbatim, just uniformly replacing ``$\alpha$'' with~``2''.

\begin{proof}[Proof of Theorem~\ref{thm:alphamergeLower}]
  Fix $\alpha > 1$. For $n \ge 1$, we define a sequence $\Ramerge(n)$ of run
  lengths that will establish the lower bound. Define $N_0$ to equal
  $3\cdot\lceil\alpha +  1\rceil$. For $n < N_0$, set $\Ramerge$ to be the
  sequence $\langle n \rangle$, containing a single run of length~$n$.
  For $n \ge N_0$, define $n^\ppprime = \lfloor \frac n{\alpha+1} \rfloor + 1$
  and $n^* = n - n^\ppprime$. Thus $n^\ppprime$ is the least integer greater than $n / (\alpha +  1)$. Similarly define
  $n^\pprime = \lfloor \frac {n^*}{\alpha + 1} \rfloor + 1$ and
  $n^\prime = n^* - n^\pprime$. These four values can be equivalently uniquely
  characterized as satisfying
  \begin{equation}\label{eq:npppnstar}
    \textstyle
    n^\ppprime = \frac 1{\alpha+1} n + \epsilon_1
    \qquad\qquad \hbox{and} \qquad\qquad\qquad
    n^* = \frac {\alpha}{\alpha+1} n - \epsilon_1
  \end{equation}
  \begin{equation}\label{eq:nppnp}
    \textstyle
    n^\pprime = \frac \alpha{(\alpha+1)^2} n
                 - \frac{1}{\alpha+1}\epsilon_1 + \epsilon_2
    \qquad \hbox{and} \qquad\qquad
    n^\prime = \frac {\alpha^2}{(\alpha+1)^2}n
                 - \frac \alpha{\alpha+1}\epsilon_1 - \epsilon_2
  \end{equation}
  for some $\epsilon_1, \epsilon_2 \in (0,1]$.
  The sequence $\Ramerge(n)$ of run
  lengths is inductively defined to be the concatenation
  of $\Ramerge(n^\prime)$, $\Ramerge(n^\pprime)$ and $\Ramerge(n^\ppprime)$.

  \begin{claim} \label{clm:alphemergeLowerNprimes}
    Let $n \ge N_0$.
    \begin{itemize}
      \setlength{\topsep}{0pt}
      \setlength{\itemsep}{0pt}
      \item[\rm a.] $n = n^\prime + n^\pprime + n^\ppprime$ and
        $n^* = n^\prime + n^\pprime$.
      \item[\rm b.] $\alpha (n^\ppprime - 3) \le n^*  < \alpha n^\ppprime$.
      \item[\rm c.] $\alpha (n^\pprime - 3) \le n^\prime < \alpha n^\pprime$.
      \item[\rm d.] $n^\ppprime \ge 3$.
      \item[\rm e.] $n^\prime \ge 1$ and $n^\pprime \ge 1$.
    \end{itemize}
  \end{claim}

  Part~a.\ of the claim is immediate from the definitions. Part~b.\ is immediate
  from the equalities~(\ref{eq:npppnstar}) since $0< \epsilon_1\le 1$ and
  $\alpha > 1$. Part~c.\ is similarly immediate from~(\ref{eq:nppnp}) since
  also $0 < \epsilon_2 \le 1$. Part~d.\ follows from (\ref{eq:npppnstar}) and
  $n \ge N_0 \ge 3 (\alpha + 1)$. Part~e. follows by (\ref{eq:nppnp}),
  $\alpha > 1$, and $n \ge N_0$.

  \begin{claim}\label{clm:Ramerge}
    Let $\Ramerge(n)$ be as defined above.
    \begin{itemize}
      \setlength{\topsep}{0pt}
      \setlength{\itemsep}{0pt}

      \item[\rm a.] The sums of the run lengths in $\Ramerge(n)$ is $n$.
      \item[\rm b.] If $n \ge N_0$, then the final run length in $\Ramerge(n)$
        is $\ge 3$.
      \item[\rm c.] Suppose that $\Ramerge(n)$ is the initial subsequence of a
        sequence~$\mathcal R^\prime$ of run lengths and that
        $\alpha$-merge sort is initially started with run lengths~$\mathcal
        R^\prime$ and (a)~with the stack~$\mathcal \QQ$ empty or (b)~with the
        top element of~$\mathcal \QQ$ a run of length $\ge \alpha (n - 3)$.
        Then $\alpha$-merge sort will start by processing exactly the runs
        whose lengths are those of $\Ramerge(n)$, merging them into single
        run which becomes the new top element of~$\mathcal \QQ$. This will be
        done without merging any runs that were initially in~$\mathcal \QQ$
        and without (yet) processing any of the remaining runs
        in~$\mathcal{R}^\prime$.
      \item[\rm d.] The property~c.\ also holds for $\alpha$-stack sort.
    \end{itemize}
  \end{claim}

  Part~a.\ is immediate from the definitions using induction on~$n$. Part~b.\
  is a consequence of Claim~\ref{clm:alphemergeLowerNprimes}(d.)\  and the fact
  that the final entry of $\Ramerge$ is a value $n^\ppprime<N_0$ for some~$n$.
  Part~c.\ is proved by induction on~$n$, similarly to the proof of
  Claim~\ref{clm:timsortLower}.  It is trivial for the base case $n<N_0$. For
  $n\ge N_0$, $\Ramerge(n)$~is the concatenation of
  $\Ramerge(n^\prime), \Ramerge(n^\pprime), \Ramerge(n^\ppprime)$. Applying the
  induction hypothesis to $\Ramerge(n^\prime)$ yields that these runs are
  initially merged into a single new run of length~$n^\prime$ at the top of the
  stack. Then applying the induction hypothesis to $\Ramerge(n^\pprime)$ shows
  that those runs are merged to become the top run on the stack. Since the last
  member of $\Ramerge(n^\pprime)$ is a run of length $\ge 3$, every intermediate
  member placed on the stack while merging the runs of $\Ramerge(n^\pprime)$ has
  length $\le n^\pprime - 3$. And, by
  Claim~\ref{clm:alphemergeLowerNprimes}(c.),
  these cannot cause a merge with the run of length~$n^\prime$ already
  in~$\mathcal \QQ$. Next, again by Claim~\ref{clm:alphemergeLowerNprimes}(c.),
  the top two members of the stack are merged to form a run of length~$n^* =
  n^\prime +  n^\pprime$.  Applying the induction hypothesis a third time, and
  arguing similarly with Claim~\ref{clm:alphemergeLowerNprimes}(b.), gives that
  the runs of $\Ramerge(n^\ppprime)$ are merged into a single run of
  length~$n^\ppprime$, and then merged with the run of length~$n^*$ to obtain a
  run of length~$n$.  This proves part~c.\ of Claim~\ref{clm:Ramerge}. Part~d.\
  is proved exactly like part~c.; the fact that $\alpha$-stack sort is only
  2-aware and never merges $X$ and~$Y$ makes the argument slightly easier in
  fact.  This completes the proof of Claim~\ref{clm:Ramerge}.

  \begin{claim}\label{clm:amergeRecur1}
    Let $c(x)$ equal the merge cost of $\alpha$-merge sort on an input
    sequence with run lengths given by $\Ramerge(n)$. Then $c(n) = 0$ for
    $n < N_0$. For $n\ge N_0$,
    \begin{eqnarray}\label{eq:camergeRecur}
      c(n) &=& c(n^\prime) + c(n^\pprime) + c(n^\ppprime)
                 + 2n^\prime + 2n^\pprime + n^\ppprime \\
           &=& c(n^\prime) + c(n^\pprime) + c(n^\ppprime)
                 + n + n^\prime + n^\pprime.
    \end{eqnarray}
    For $n\ge N_0$, $c(n)$ is strictly increasing as a function of~$n$.
  \end{claim}
  The first equality of Equation~(\ref{eq:camergeRecur}) is an immediate consequence of the
  proof of part~c.\ of Claim~\ref{clm:Ramerge};
  the second follows from $n = n^\prime + n^\pprime + n^\ppprime$. To see that $c(n)$ is
  increasing for $n\ge N_0$,
  let $(n +\penalty10000 1)^\prime, (n +\penalty10000 1)^\pprime, (n +\penalty10000 1)^\ppprime$
  indicate the three values such that
  $\Ramerge(n +  1)$ is the concatenation of
  $\Ramerge((n +  1)^\prime)$, $\Ramerge((n +  1)^\pprime)$ and
  $\Ramerge((n +  1)^\ppprime)$. Note that
  $(n +  1)^\prime \ge n^\prime$,
  $(n +  1)^\pprime \ge n^\pprime$, and
  $(n +  1)^\ppprime \ge n^\ppprime$.
  An easy proof by induction now shows that $c(n + 1) > c(n)$ for $n \ge N_0$,
  and Claim~\ref{clm:amergeRecur1} is proved.

  Let
  $\delta = \lceil 2(\alpha + 1)^2/(2 \alpha + 1)\rceil$.
  (For $1 < \alpha \le 2$, we have $\delta \le 4$.)
  We have $\delta \le N_0 - 1$
  for all $\alpha > 1$.
  For real $x\ge N_0$,
  define $b(x) = c(\lfloor x \rfloor  -  \delta)$.
  Since $c(n)$ is increasing, $b(x)$ is nondecreasing.

  \begin{claim}\label{clm:bamergeRecur}
      \begin{itemize}
          \setlength{\parsep}{0pt}
          \setlength{\itemsep}{0pt}
          \item[\rm a.]
            $\frac 1{\alpha + 1} n - \delta \le (n - \delta)^\ppprime$.
          \item[\rm b.]
            $\frac \alpha{(\alpha + 1)^2} n - \delta \le (n - \delta)^\pprime$.
          \item[\rm c.]
            $\frac {\alpha^2}{(\alpha + 1)^2} n - \delta \le
              (n - \delta)^\prime$.
          \item[\rm d.] If $x \ge N_0 + \delta$, then
                $b(x) \, \ge\,
                b({\textstyle \frac{\alpha^2}{(\alpha + 1)^2} x }) +
                b({\textstyle \frac\alpha{(\alpha + 1)^2} x }) +
                b({\textstyle \frac1{\alpha + 1} x }) +
                {\textstyle \frac{2 \alpha + 1}{\alpha + 1} }
                  (x  -  \delta - 1) - 1$.
      \end{itemize}
  \end{claim}
  For~a.,
  (\ref{eq:npppnstar}) implies that
  $(n - \delta)^\ppprime \ge \frac{n-\delta}{\alpha+1}$,
  so a.~follows from $-\delta \le -\delta/(\alpha +  1)$.
  This holds as $\delta > 0$ and $\alpha>1$. For~b.,
  (\ref{eq:nppnp})~implies that
  $(n - \delta)^\pprime \ge
  \frac{\alpha}{(\alpha + 1)^2}(n - \delta) - \frac{1}{\alpha + 1}$,
  so after simplification, b.~follows from
  $\delta \ge (\alpha + 1) / (\alpha^2 +  \alpha +  1)$;
  it is easy to verify that this holds by choice of~$\delta$.
  For~c., (\ref{eq:nppnp})~also implies that
  $(n - \delta)^\prime \ge \frac{\alpha^2}{(\alpha + 1)^2}(n - \delta) - 2$,
  so after simplification, c.~follows from
  $\delta \ge 2(\alpha +  1)^2 / (2\alpha +  1)$.

  To prove part~d., letting $n = \lfloor x \rfloor$
  and using parts a., b.\ and c., equations
  (\ref{eq:npppnstar}) and (\ref{eq:nppnp}),
  and the fact that $b(x)$ and $c(n)$ are nondecreasing, we have
  \begin{eqnarray*}
    b(x) &=& c( n - \delta ) \\
      &=& c( (n - \delta)^\prime ) +
         c( (n - \delta)^\pprime ) +
         c( (n - \delta)^\ppprime ) +
         (n -\delta) + (n - \delta)^\prime + (n - \delta)^\pprime \\
      &\ge& \textstyle
         c( \lfloor \frac{\alpha^2}{(\alpha+1)^2} n \rfloor - \delta ) +
         c( \lfloor \frac \alpha {(\alpha+1)^2} n \rfloor - \delta ) +
         c( \lfloor \frac 1 {\alpha+1} n \rfloor - \delta ) +
      \\ & & \quad\quad\quad +
         \textstyle
         \bigl( 1 + \frac{\alpha^2}{(\alpha+1)^2} +
                \frac {\alpha}{(\alpha+1)^2} \bigr) \cdot (n - \delta)
      \\ & & \quad\quad\quad
         \textstyle
        - \frac {\alpha}{\alpha+1} \epsilon_1 -  \epsilon_2
        - \frac 1{\alpha+1}\epsilon_1 +  \epsilon_2 \\
      &\ge& \textstyle
        b({\textstyle \frac{\alpha^2}{(\alpha+1)^2} x }) +
        b({\textstyle \frac\alpha{(\alpha+1)^2} x }) +
        b({\textstyle \frac1{\alpha+1} x }) +
        \frac{ 2 \alpha+1 }{\alpha+1} (x - \delta - 1) - 1.
  \end{eqnarray*}
  Claim~\ref{clm:bamergeRecur}(d.) gives us the basic recurrence needed
  to lower bound $b(x)$ and hence $c(n)$.

  \begin{claim}\label{clm:bamergeLowerBd}
    For all $x \ge \delta + 1$,
    \begin{equation}\label{eq:bamergeLower}
      b(x)~\ge~ c_\alpha\cdot x \log x -Bx +A,
    \end{equation}
    where $A = \frac {2\alpha + 1}{2\alpha + 2}(\delta +  1) + \frac12$ and
    $B = \frac A {\delta + 1} +
    c_\alpha \log (\max\{
      N_0 + \delta +  1,
      \left\lceil \frac{(\delta + 1)(\alpha + 1)^2}{\alpha} \right\rceil
    \})$.
  \end{claim}

  The claim is proved by induction, namely we prove by
  induction on~$n$ that (\ref{eq:bamergeLower})
  holds for all $x < n$. The base case is for
  $x < n = \max\{ N_0 +  \delta +  1,
                  \left\lceil \frac{(\delta + 1)(\alpha + 1)^2}{\alpha} \right\rceil
               \}$.
  In this case $b(x)$ is non-negative, and the righthand side
  of~(\ref{eq:bamergeLower}) is $\le 0$ by choice of~$B$.
  Thus (\ref{eq:bamergeLower}) holds trivially.

  For the induction step, we may assume $n-1\le x < n$ and have
  \[
  \textstyle
  \delta + 1 \le \frac\alpha{(\alpha+1)^2} x < \frac{\alpha^2}{(\alpha + 1)^2} x
      < \frac1{\alpha + 1} x < n -1.
  \]
  The first of these inequalities follows from $x \ge \frac{(\delta + 1)(\alpha + 1)^2}{\alpha}$;
  the remaining inequalities follow from $1<\alpha<2$ and $x < n$ and $n\ge N_0 \ge 6$.
  Therefore, the induction hypothesis implies that the bound~(\ref{eq:bamergeLower})
  holds for
  $b(\frac{\alpha^2}{(\alpha+1)^2} x )$,
  $b(\frac\alpha{(\alpha+1)^2} x )$,
  and $b(\frac1{\alpha+1} x)$. So by Claim~\ref{clm:bamergeRecur}(d.),
  \begin{eqnarray*}
  b(x) &\ge& \textstyle
       c_\alpha \frac{\alpha^2}{(\alpha+1)^2}
           x \log \frac {\alpha^2 x}{(\alpha+1)^2}
           - B \frac {\alpha^2}{(\alpha+1)^2} x
           + \frac {2\alpha+1}{2\alpha+2}(\delta +  1) + \frac12 \\
       && \quad \textstyle
       + c_\alpha \frac{\alpha}{(\alpha+1)^2}
           x \log \frac {\alpha x}{(\alpha+1)^2}
           - B \frac {\alpha}{(\alpha+1)^2} x
           + \frac {2\alpha+1}{2\alpha+2}(\delta +  1) + \frac12 \\
       && \quad \textstyle
       + c_\alpha \frac1{\alpha+1}
           x \log \frac{x}{\alpha+1}
           - B \frac{1}{\alpha+1} x
           + \frac {2\alpha+1}{2\alpha+2}(\delta +  1) + \frac12 \\
       && \quad \textstyle
           + \frac{2 \alpha + 1}{\alpha+1} x
           - \frac{2 \alpha + 1}{\alpha+1} (\delta +  1) - 1 \\
   &=& \textstyle
       c_\alpha x \log x - B x + A \\
   && \quad \textstyle
        + c_\alpha x \left[
              \frac{\alpha^2}{(\alpha+1)^2} \log\frac{\alpha^2}{(\alpha+1)^2}
              + \frac{\alpha}{(\alpha+1)^2} \log\frac{\alpha}{(\alpha+1)^2}
              + \frac1{\alpha+1} \log\frac1{\alpha+1} \right]
              + \frac{2 \alpha + 1}{\alpha+1} x .
  \end{eqnarray*}
  The quantity in square brackets is equal to
  \begin{eqnarray*}
  \lefteqn{ \textstyle
        \left( \frac{2 \alpha^2}{(\alpha+1)^2}
               + \frac{\alpha}{(\alpha+1)^2} \right)\log \alpha
        - \left( \frac{2 \alpha^2}{(\alpha+1)^2}
               + \frac{2 \alpha}{(\alpha+1)^2}
               + \frac1{\alpha+1} \right) \log(\alpha +  1)} \\
  &=& \textstyle
        \frac {\alpha(2\alpha+1)}{(\alpha+1)^2} \log \alpha
        - \frac {2\alpha^2+3\alpha+1}{(\alpha+1)^2} \log(\alpha +  1)
        \hspace*{1in} 
        \\
  &=& \textstyle
        \frac {\alpha(2\alpha+1)}{(\alpha+1)^2} \log \alpha
        - \frac {2\alpha+1}{\alpha+1} \log(\alpha +  1) \\
  &=& \textstyle
        \frac {\alpha \log \alpha - (\alpha +  1) \log(\alpha +  1) } {\alpha+1}
        \cdot \frac{2 \alpha + 1}{\alpha+1}.
  \end{eqnarray*}
  Since
  $c_\alpha = (\alpha +  1) /
    ((\alpha +  1) \log(\alpha +  1) - \alpha \log \alpha)$,
  we get that $b(x) \ge c_\alpha x \log x - B x + A$.
  This completes the induction step and proves Claim~\ref{clm:bamergeLowerBd}.

  Since $A$ and $B$ are constants, this gives
  $b(x) \ge (c_\alpha - o(1)) x \log x$.
  Therefore $c(n) = b(n +  \delta) \ge (c_\alpha - o(1)) n \log n$.
  This completes the proofs of Theorems \ref{thm:alphaLower},
  \ref{thm:alphamergeLower} and~\ref{thm:twomergeLower}.
\end{proof}

\subsection{Upper bound for 2-merge sort and \texorpdfstring{$\alpha$}{a}-merge sort --- preliminaries}
\label{sec:mergeUpperPrelim}

We next prove upper bounds on the worst-case runtime
of $2$-merge sort and $\alpha$-merge sort for $\varphi<\alpha < 2$.
The upper bounds will have the form
$n \cdot ( d_\alpha + c_\alpha \log n)$, with no hidden or missing constants.
$c_\alpha$~was already defined in~(\ref{eq:cAlphaDef}). For $\alpha =2$, $c_\alpha \approx 1.08897$
and the constant~$d_\alpha$ is
\begin{equation}\label{eq:dtwoDef}
  d_2 ~=~ 6 - c_2\cdot( 3 \log 6 -2 \log 4)
      ~=~ 6 - c_2\cdot( (3 \log 3) - 1) ~\approx~ 1.91104.
\end{equation}
For $\varphi<\alpha<2$, first define
\begin{equation}\label{eq:kalphaDef}
  k_0(\alpha) ~=~ \min \{
    {
      \textstyle \ell \in N ~:~
      \frac{\alpha^2 - \alpha - 1}{\alpha - 1} \ge \frac{1}{\alpha^\ell}
    }
  \}.
\end{equation}
Note $k_0(\alpha)\ge 1$. Then set, for $\varphi<\alpha<2$,
\begin{equation}\label{eq:dalphaDef}
  d_\alpha ~=~
  \frac{2^{k_0(\alpha)+1} \cdot \max\{ (k_0(\alpha) +  1), 3 \} \cdot
    (2\alpha - 1)}{\alpha - 1} + 1 .
\end{equation}
Our proof for $\alpha = 2$ is substantially simpler than the
proof for general $\alpha$: it also gives the better constant~$d_2$.
The limits $\lim_{\alpha\rightarrow \varphi^+} k(\alpha)$
and $\lim_{\alpha\rightarrow \varphi^+} d_\alpha$ are both equal
to~$\infty$; we suspect this is not optimal.\footnote{%
Already the term $2^{k_0(\alpha)+1}$ is not optimal as the
proof of Theorem~\ref{thm:alphamergeUpper} shows that the base~2 could be
replaced by $\sqrt{\alpha +  1}$; we conjecture however, that
in fact it is not necessary for the limit of~$d_\alpha$ to be
infinite.} However, by Theorems \ref{thm:alphamergeLower} and~\ref{thm:twomergeLower}, the constant~$c_\alpha$ is optimal.

\begin{theorem}\label{thm:alphamergeUpper}
  Let $\varphi < \alpha < 2$. The merge cost of $\alpha$-merge sort on
  inputs of length~$n$ composed of $m$ runs is
  $\le n \cdot (d_\alpha + c_\alpha \log m)$.
\end{theorem}

\noindent
The corresponding upper bound for $\alpha=2$ is:

\begin{theorem}\label{thm:twomergeUpper}
  The merge cost of $2$-merge sort on inputs of length~$n$ composed of $m$
  runs is $\le n \cdot (d_2 + \penalty10000 c_2\log m) \approx
    n \cdot (1.91104 + 1.08897 \log m)$.
\end{theorem}

Proving Theorems \ref{thm:alphamergeUpper} and \ref{thm:twomergeUpper}
requires handling three situations:
First, the algorithm may have top stack element~$Z$ which is
not too much larger than~$Y$ (so $|Z| \le \alpha |Y|$): in this
case either $Y$ and~$Z$ are merged or $Z$ is small
compared to~$Y$ and no merge occurs. This first case
will be handled by case~(A) of the proofs. Second,
the top stack element~$Z$ may be much larger than~$Y$
(so $|Z| > \alpha |Y|$): in this case, the algorithm will
repeatedly merge $X$ and~$Y$ until $|Z| \le |X|$. This is the
most complicated case of the argument, and is handled
by cases (C) and (D) of the proofs. (Case~(D)
is not needed when $\alpha = 2$.)
In the third case, the original runs in~$\mathcal R$
have been exhausted and the final loop on lines
\ref{algline:amergefinalloopA}-\ref{algline:amergefinalloopB}
repeatedly merges $Y$ and~$Z$. The third case is handled by
case~(B) of the proofs.

The next four technical lemmas are key for the proofs.
Lemma~\ref{lem:alphamerge-A} is used for case~(A) of
the proofs; the constant~$c_\alpha$ is exactly what
is needed to make this hold. Lemma~\ref{lem:alphamerge-B}
is used for case~(B) of the proofs.
Lemma~\ref{lem:alphamerge-Ctwo}
is used for case~(C),
and Lemma~\ref{lem:alphamerge-CD} is used
in case~(D) when $\alpha < 2$.

\begin{lemma}\label{lem:alphamerge-A}
Let $\alpha>1$.
Let $A,B,a,b$ be positive integers such that
$A \le \alpha B$ and $B \le \alpha A$.  Then
\begin{equation}\label{eq:alphamergebd1}
A \cdot c_\alpha \log a + B \cdot c_\alpha\log b + A + B
~\le~
(A  +   B) \cdot c_\alpha \log(a  +   b).
\end{equation}
\end{lemma}

\begin{lemma}\label{lem:alphamerge-B}
  Let $\alpha > 1$. Let $A, B, a, b$ be positive integers such that
  $(\alpha - 1)B \le A$.  Then
  \begin{equation}\label{eq:alphamergebd2}
    A \cdot c_\alpha \log a + B \cdot (1 + c_\alpha \log b) + A + B
    ~\le~
    (A + B) \cdot (1 + c_\alpha \log (a  +   b)).
  \end{equation}
\end{lemma}

\begin{lemma}\label{lem:alphamerge-Ctwo}
  Let $\varphi<\alpha\le 2$ and $A$, $B$, $a$, $b$ be positive integers.
  \begin{itemize}
    \item[\rm (a)] (For $\alpha=2$.) If $a\ge 2$ and $A \le 2B$, then
      \[
        A \cdot (d_2 + c_2 \log (a - 1)) + A + B
          \le (A  +   B) \cdot (d_2 - 1 + c_2 \log (a +  b)) .
      \]
    \item[\rm (b)] If $\varphi<\alpha<2$ and
      $A \le \frac{\alpha}{\alpha - 1} \cdot B$, then
      \[
        A \cdot (d_\alpha + c_\alpha \log a) + A + B
          \le (A + B) \cdot (d_\alpha - 1 + c_\alpha \log (a + b)) .
      \]
  \end{itemize}
\end{lemma}

\begin{lemma}\label{lem:alphamerge-CD}
  Let $\varphi < \alpha < 2$. Let $A$, $B$, $C$, $a$, $b$, and $k$ be positive integers such that $k\le k_0(\alpha)+1$ and
  $\frac{2^k (2\alpha - 1)}{\alpha - 1} C \ge A + B + C$. Then
  \begin{multline}\label{eq:alphamerge-CD}
      A \cdot (d_\alpha + c_\alpha \log a) +
      B \cdot (d_\alpha + c_\alpha \log b)) +
      k \cdot C + 2 B + A
      ~\le~\\
      A \cdot(d_\alpha - 1 + c_\alpha \log a) +
      (B + C) \cdot(d_\alpha - 1 + c_\alpha \log (b  +   1)).
  \end{multline}
\end{lemma}

\begin{proof}[Proof of Lemma~\ref{lem:alphamerge-A}]
  The inequality~(\ref{eq:alphamergebd1}) is equivalent to
  \[
    A \cdot (c_\alpha \log a - c_\alpha \log(a  +   b) + 1)
    ~\le~
    B \cdot (c_\alpha \log(a + b) - c_\alpha \log b - 1)
  \]
  and hence to
  \[
    A \cdot \Bigl(1 + c_\alpha \log \frac{a}{a + b} \Bigr)
    ~\le~
    B \cdot \Bigl(-1 - c_\alpha \log \frac{b}{a + b} \Bigr).
  \]
  Setting $t = b / (a  +   b)$, this is the same as
  \begin{equation}\label{eq:alphamergebd1_1}
    A \cdot(1 + c_\alpha \log (1  -  t))
    ~\le~
    B \cdot(-1 - c_\alpha \log t).
  \end{equation}
  Let $t_0 = 1 - 2^{-1 / c_\alpha}$. Since $c_\alpha>1$, we have $t_0 < 1/2$,
  so $t_0 < 1  -  t_0$. The lefthand side of~(\ref{eq:alphamergebd1_1}) is
  positive iff $t < t_0$. Likewise, the righthand side is positive
  iff $t < 1 - t_0$. Thus (\ref{eq:alphamergebd1_1}) certainly
  holds when $t_0 \le t \le 1 - t_0$ where the lefthand side is
  $\le 0$ and the righthand side is $\ge 0$.

  Suppose $0<t<t_0$, so $1 + c_\alpha \log (1  -  t)$
  and $-1 - c_\alpha \log t$ are both positive.
  Since $A\le \alpha B$, to prove~(\ref{eq:alphamergebd1_1})
  it will suffice to prove
  $\alpha (1 + c_\alpha \log(1  -  t)) \le -1 - c _\alpha\log t$,
  or equivalently that
  $-1 - \penalty10000 \alpha \ge c_\alpha \log( t (1  -  t)^\alpha )$. The derivative
  of $\log( t (1  -  t)^\alpha )$ is
  $(1- (1 + \alpha)t)/(t(1 - t))$; so $\log( t (1  -  t)^\alpha )$
  is maximized at $t = 1/(1  +   \alpha)$ with value
  $\alpha \log \alpha - (\alpha + 1)\log(\alpha + 1)$. Thus
  the desired inequality holds by the definition of~$c_\alpha$, so
  (\ref{eq:alphamergebd1_1}) holds in this case.

  Now suppose $1 - t_0 < t < 1$, so $1 + c_\alpha \log (1  -  t)$
  and $-1 - c_\alpha \log t$ are both negative. Since $B\le\alpha A$,
  it suffices to prove
  $1 + c_\alpha \log(1  -  t) \le \alpha (-1 - c_\alpha \log t)$
  or equivalently that
  $-1 - \alpha \ge c_\alpha \log( t^\alpha (1  -  t))$.
  This is identical to the situation of the previous paragraph, but
  with $t$ replaced by $1  -  t$, so
  (\ref{eq:alphamergebd1_1}) holds in this case also.
\end{proof}

\begin{proof}[Proof of Lemma~\ref{lem:alphamerge-B}]
  The inequality (\ref{eq:alphamergebd2}) is equivalent to
  \[
    B \cdot (1 + c_\alpha \log b - c_\alpha \log(a  +   b))
    ~\le~
    A \cdot (c_\alpha \log(a  +   b) - c_\alpha \log a).
  \]
  Since $\log(a  +   b) - \log a > 0$ and $(\alpha  -  1)B \le A$,
  it suffices to prove
  \[
    1 + c_\alpha \log b - c_\alpha \log(a  +   b)
      ~\le~
      (\alpha - 1) \cdot (c_\alpha \log(a  +   b) - c_\alpha \log a).
  \]
  This is equivalent to
  \[
    c_\alpha\log b + (\alpha  -  1)\cdot c_\alpha\log(a) - \alpha\cdot c_\alpha \log(a +  b)
      ~\le~ -1.
  \]
  Letting $t = b / (a +  b)$, we must show
  $-1 \ge c_\alpha\log(t (1  -  t)^{\alpha - 1})$.
  Similarly to the previous proof, taking the first derivative shows that
  the righthand side is maximized with $t = 1/\alpha$, so it will suffice
  to show that
  $-1 \ge c_\alpha \log( (\alpha  -  1)^{\alpha-1}/\alpha^\alpha )$,
  i.e., that
  $c_\alpha\cdot(\alpha \log \alpha - (\alpha  -  1)\log(\alpha  -  1)) \ge 1$.
  Numerical examination shows that this is true for $\alpha>1.29$.
\end{proof}

\begin{proof}[Proof of Lemma~\ref{lem:alphamerge-Ctwo}]
  We assume w.l.o.g.\ that $b = 1$.
  The inequality of part~(a) is equivalent to
  \[
      A \cdot (2 + c_2 \log (a - 1) - c_2 \log (a +  1))
      ~\le~
      B \cdot (d_2 - 2 + c_2 \log (a  +   1)).
  \]
  Since the righthand side is positive and $A \le 2 B$,
  we need to prove
  \[
    2 \cdot(2 + c_2 \log(a - 1) - c_2 \log(a +  1))
    ~\le~
    d_2 - 2 + c_2 \log (a +  1).
  \]
  This is easily seen to be the same as
  $6 - d_2 \le
      c_2 \log ((a  +   1)^3/(a - 1)^2)$.
  With $a>1$ an integer, the quantity
  $(a  +   1)^3/(a - 1)^2$ is minimized when $a=5$.
  Thus, we must show that
  \[
    d_2 ~\ge~ 6 - c_2 \log( 6^3/4^2 ) ~=~ 6 - c_2(3 \log 6 - 2 \log 4).
  \]
  In fact, $d_2$ was defined so that equality holds.
  Thus (a) holds.

  Arguing similarly for part~(b), we must show
  \[
      \frac{2\alpha}{\alpha-1}+2-d_\alpha
      ~\le~ c_2 \log ((a  +   1)^{\alpha/(\alpha-1)}/a).
  \]
  This holds trivially, as $d_\alpha\ge 12$
  for $\alpha<2$ so the lefthand side is negative and
  the righthand side is positive.
\end{proof}

\begin{proof}[Proof of Lemma~\ref{lem:alphamerge-CD}]
    The inequality~(\ref{eq:alphamerge-CD}) is equivalent to
    \[
      2 A + B \cdot (3 + c_\alpha \log b - c_\alpha \log (b  +   1)) + k \cdot C
      ~\le~
      C \cdot (d_\alpha - 1 + c_\alpha \log (b  +   1)).
    \]
    Since $\frac{2^k (2\alpha - 1)}{\alpha - 1} C \ge  A + B + C$, it is enough
    to prove
    \[
      \max(k, 3) \cdot \frac{2^k (2\alpha  -  1)}{\alpha  -  1} C
      ~\le~
      (d_\alpha  -  1) \cdot C .
    \]
    This is equivalent to
    \[
      d_\alpha ~\ge~
      \frac{2^k \cdot \max\{k, 3\} \cdot (2\alpha  -  1)}{\alpha - 1} + 1.
    \]
    This holds by the definition of $d_\alpha$,
    since $k \le k_0(\alpha) +  1$.
\end{proof}

The proofs of Theorems \ref{thm:alphamergeUpper} and~\ref{thm:twomergeUpper}
use two functions $G_\alpha$ and~$H_\alpha$ to bound the
merge cost of runs stored on the stack.
\begin{definition}
  For $\alpha=2$, define
  \begin{eqnarray*}
    G_2(n,m) &=& n \cdot (d_2 - 1 + c_2 \log m) \\[1ex]
    H_2(n,m) &=&
    \left\{
      \begin{array}{ll}
        n \cdot (d_2 + c_2 \log (m - 1)) \quad & \hbox{if $m\ge 2$} \\
        0 & \hbox{if $m=1$.}
      \end{array}
    \right.
  \end{eqnarray*}
  For $\alpha < 2$,
  \begin{eqnarray*}
      G_\alpha(n,m) &=& n \cdot (d_\alpha - 1 + c_\alpha \log m) \\[1ex]
      H_\alpha(n,m) &=& n \cdot (d_\alpha     + c_\alpha \log m).
  \end{eqnarray*}
\end{definition}
Recall that $m_X$ is the number of original runs
merged to form a run~$X$.
For the proof of Theorem~\ref{thm:twomergeUpper} in the next section,
upper bounding the merge cost of $2$-merge sort,
the idea is that for most runs~$X$ on the stack~$\mathcal \QQ$,
the merge cost of~$X$ will be bounded by $G_2(|X|, m_X)$.
However, many of the runs formed by merges in cases (B) and~(C)
will instead have merge cost bounded by $H_2(|X|,m_X)$.

A similar intuition applies to the proof of Theorem~\ref{thm:alphamergeUpper}
for $\varphi<\alpha<2$, in Section~\ref{sec:alphaMergeUpperPf}.
However, the situation is more complicated as
that proof will bound the total merge cost
instead of individual merge costs~$w_{\QQ_i}$.

The next lemma is the crucial property of $G_2$ and~$G_\alpha$
that is needed for both Theorems
\ref{thm:twomergeUpper} and~\ref{thm:alphamergeUpper}.
The lemma is used to bound the merge costs incurred when merging
two runs which differ in size by at most a factor~$\alpha$.
The constant $c_\alpha$ is exactly what is needed to make this lemma hold.
\begin{lemma}\label{lem:Gbound}
  Suppose $n_1, n_2, m_1, m_2$ are positive
  integers, and $\varphi < \alpha \le 2$. Also suppose $n_1 \le \alpha n_2$ and
  $n_2 \le \alpha n_1$. Then,
  \[
    G_\alpha(n_1, m_1) + G_\alpha(n_2, m_2) + n_1 + n_2
    ~\le~
    G_\alpha(n_1  +   n_2, m_1  +   m_2).
  \]
\end{lemma}
\begin{proof}
  The inequality expresses that
  \[
    n_1 \cdot (d_\alpha - 1 + c_\alpha \log m_1)
     + n_2 \cdot (d_\alpha - 1 + c_\alpha \log m_2) + n_1 + n_2
    ~\le~
    (n_1 + n_2) \cdot (d_\alpha - 1 + c_\alpha \log (m_1  +   m_2)).
  \]
  This is an immediate consequence of Lemma~\ref{lem:alphamerge-A} with $A$,
  $B$, $a$, $b$ replaced with $n_1$, $n_2$, $m_1$, $m_2$.
\end{proof}

\subsection{Upper bound proof for 2-merge sort}
\label{sec:twoMergeUpperPf}
This section gives the proof of Theorem~\ref{thm:twomergeUpper}.
Lemma~\ref{lem:GHtwobounds} states some properties of $G_2$ and $H_2$ which
follow from Lemmas \ref{lem:alphamerge-B} and~\ref{lem:alphamerge-Ctwo}(a).

\begin{lemma}\label{lem:GHtwobounds}
  Suppose $n_1, n_2, m_1, m_2$ are positive integers.
  \begin{description}
    \item[\rm (a)] If $n_2 \le n_1$, then
      \[
        G_2(n_1, m_1) + H_2(n_2, m_2) + n_1 + n_2
        ~\le~
        H_2(n_1  +   n_2, m_1  +   m_2).
      \]
    \item[\rm (b)]
      If $n_1 \le 2 n_2$, then
      \[
        H_2(n_1,m_1) + n_1 + n_2 ~\le~ G_2(n_1 +  n_2, m_1 +  m_2) .
      \]
  \end{description}
\end{lemma}

\begin{proof}
  If $m_2 \ge 2$, part~(a) states that
  \[
    n_1 \cdot (d_2 - 1 + c_2 \log m_1)
     + n_2 \cdot (d_2 + c_2 \log (m_2 - 1)) + n_1 + n_2
    ~\le~
    (n_1  +   n_2) \cdot (d_2 + c_2 \log (m_2  +   m_2  -  1)).
  \]
  This is an immediate consequence of Lemma~\ref{lem:alphamerge-B}.
  If $m_2 = 1$, then part~(a) states
  \[
    n_1 \cdot (d_2 - 1 + c_2 \log m_1) + n_1 + n_2
    ~\le~
    (n_1 + n_2) \cdot (d_2 + c_2 \log m_1).
  \]
  This holds since $d_2 \ge 1$ and $n_2 > 0$ and $m_1 \ge 1$.

  When $m_1 \ge 2$, part~(b) states that
  \[
    n_1 \cdot (d_2 + c_2 \log (m_1 - 1)) + n_1 + n_2
    ~\le~
    (n_1 + n_2) \cdot (d_2 - 1 + c_2 \log (m_1  +   m_2));
  \]
  this is exactly Lemma~\ref{lem:alphamerge-Ctwo}(a).
  When $m_1 = 1$, (b)~states
  $n_1 +  n_2 \le (n_1 +  n_2)(d_2  -  1  +   c_2\log(m_2 +  1))$,
  and this is trivial since $c_2 + d_2 \ge 2$ and $m_2 \ge 1$.
\end{proof}

We next prove Theorem~\ref{thm:twomergeUpper}. We use the convention that
the $2$-merge sort algorithm maintains a stack~$\mathcal \QQ$ containing runs $\QQ_1$,
$\QQ_2$, \dots, $\QQ_\ell$. The last four runs are denoted $W$, $X$, $Y$, $Z$. Each
$\QQ_i$ is a run of $|\QQ_i|$ many elements. Recall that $m_{\QQ_i}$ and
$w_{\QQ_i}$ denote the number of original runs that were merged to form $\QQ_i$
and the merge cost of~$\QQ_i$ (respectively). If $\QQ_i$ is an original run, then
$m_{\QQ_i} = 1$ and $w_{\QQ_i} = 0$. If $m_{\QQ_i}=2$, then $\QQ_i$ was formed by a
single merge, so $w_{\QQ_i} = |\QQ_i|$. To avoid handling the special cases for
$\ell \le 2$, we adopt the convention that there is a virtual initial run $\QQ_0$
with infinite length, so $|\QQ_0| = \infty$.

\begin{lemma}\label{lem:GHtwobound}
  Suppose $\QQ_i$ is a run in~$\mathcal \QQ$ and that
  $w_{\QQ_i} \le G_2(|\QQ_i|, m_{\QQ_i})$. Then $w_{\QQ_i} \le H_2(|\QQ_i|, m_{\QQ_i})$.
\end{lemma}

\begin{proof}
  If $m_{\QQ_i}=1$, the lemma
  holds since $w_X = 0$. If $m_{\QQ_i} = 2$, then it holds since
  $w_{\QQ_i} = |\QQ_i|$. If $m_{\QQ_i} > 2$, then it holds since
  $c_2 \log(m_{\QQ_i} / (m_{\QQ_i} - 1)) < 1$ and hence
  $G_2(|\QQ_i|, m_{\QQ_i}) < H_2(|\QQ_i|, m_{\QQ_i})$.
\end{proof}

\begin{proof}[Proof of Theorem~\ref{thm:twomergeUpper}]
  We describe the $2$-merge algorithm by using
  three invariants (A), (B), (C) for the stack;
  and analyzing what action is taken
  in each situation.
  Initially, the stack contains a single original $\QQ_1$,
  so $\ell=1$ and $m_{\QQ_1} = 1$ and $w_{\QQ_1}=0$, and case~(A)
  applies.

  \smallskip

  \noindent
  {\bf (A): Normal mode.} The stack satisfies
  \begin{itemize}
  \setlength{\itemsep}{0pt}
  \item[(A-1)] $|\QQ_i| \ge 2 \cdot |\QQ_{i + 1}|$ for all $i < \ell  -  1$.
  This includes $|X|\ge 2|Y|$ if $\ell\ge 2$.
  \item[(A-2)] $2 \cdot |Y| \ge |Z|$; i.e.\ $2 |\QQ_{\ell - 1}| \ge |\QQ_\ell|$.
  \item[(A-3)] $w_{\QQ_i} \le G_2(|\QQ_i|, m_{\QQ_i})$ for all $i \le \ell$.
  \end{itemize}
  \noindent
  If $\ell \ge 2$, (A-1) and (A-2) imply $|X|\ge |Z|$, i.e.\
  $|\QQ_{\ell - 2}| \ge |\QQ_\ell|$.
  The $\alpha$-merge algorithm does one of the following:

  \begin{itemize}
  \item If $2 |Z| \le |Y|$ and there are no more original
  runs to load, then it goes to case~(B). We claim the four
  conditions of (B) hold (see below). The condition (B-2)
  holds by $|Y| \ge 2|Z|$, and
  (B-1) and (B-3) hold by (A-1) and (A-3).  Condition (B-4)
  holds by (A-3) and Lemma~\ref{lem:GHtwobound}.
  \item If $2 |Z| \le |Y|$ and there is another original
  run to load, then the algorithm loads the next run as $\QQ_{\ell + 1}$.
  \begin{itemize}
  \item[$\circ$] If $|\QQ_{\ell + 1}| \le 2 |\QQ_\ell|$,
      then we claim that case~(A) still holds with
      $\ell$ incremented by one. In particular, (A-1) will hold
      since $|Y| \ge 2 |Z|$ is the same as
      $2 |\QQ_\ell| \le |\QQ_{\ell - 1}|$.
      Condition (A-2)
      will hold by the assumed bound on $|\QQ_{\ell + 1}|$.
      Condition (A-3) will still hold since $|\QQ_{\ell + 1}|$ is an
      original run so $m_{\QQ_{\ell + 1}}=1$ and
      $w_{\QQ_{\ell + 1}}=0$.
  \item[$\circ$] Otherwise $|\QQ_{\ell + 1}| > 2 |\QQ_\ell|$, and
      we claim that case~(C) below holds with
      $\ell$ incremented by one. (C-1) and (C-4) will hold by (A-1) and (A-3).
      For (C-2), we need
      $|\QQ_{\ell - 1}| \ge |\QQ_\ell|$;
      i.e.\ $|Y| \ge |Z|$: this follows trivially
          from $|Y| \ge 2 |Z|$.
      (C-5) holds by (A-3) and Lemma~\ref{lem:GHtwobound}.
      (C-3) holds since $2 |\QQ_\ell| < |\QQ_{\ell + 1}|$.
      (C-6) holds since $\QQ_{\ell + 1}$ is an original run.
  \end{itemize}
  \item If $2|Z| > |Y|$, then the algorithm merges the two
  runs $Y$ and~$Z$. We claim the resulting stack satisfies condition~(A)
  with $\ell$ decremented by one. (A-1) clearly will still hold.
  For (A-2) to still hold,
  we need $2 |X| \ge |Y| +  |Z|$: this
  follows from $2 |Y| \le |X|$ and $|Z|\le |X|$.
  (A-3) will clearly still hold for all $i<\ell - 1$.
  For $i = \ell - 1$, since merging $Y$ and $Z$ added $|Y| +  |Z|$
  to the merge cost, (A-3) implies that
  the new top stack element will
  have merge cost at most
  \[
  G_2(|Y|,m_Y) + G_2(|Z|,m_Z) + |Y| + |Z|.
  \]
  By (A-2) and Lemma~\ref{lem:Gbound},
  this is $\le G_2(|Y| +  |Z|, m_Y +  m_Z)$,
  so (A-3) holds.
  \end{itemize}

  \noindent
  {\bf (B): Wrapup mode, lines 15-18 of
  Algorithm~\ref{alg:alphaMergeSort}.}
  There are no more original runs to process. The
  entire input has been combined into the
  runs $\QQ_1, \ldots, \QQ_\ell$ and they satisfy:
  \begin{itemize}
  \setlength{\itemsep}{0pt}
  \item[(B-1)] $|\QQ_i| \ge 2\cdot |\QQ_{i + 1}|$ for all $i < \ell - 1$.
  This includes $|X|\ge 2|Y|$ if $\ell\ge 2$.
  \item[(B-2)] $|Y| \ge |Z|$;
         i.e., $|\QQ_{\ell - 1}| \ge |\QQ_\ell|$.
  \item[(B-3)] $w_{\QQ_i} \le G_2(|\QQ_i|,m_{\QQ_i})$ for all $i \le \ell - 1$,
  \item[(B-4)] $w_Z \le H_2(|Z|, m_Z)$.
  \end{itemize}
  If $\ell=1$, the run $Z = \QQ_1$ contains the entire input in sorted order
  and the algorithm terminates. The total merge cost is $\le H_2(|Z|, m_Z)$.
  This is $< n \cdot (d_2 + c_2 \log m )$
  as needed for Theorem~\ref{thm:twomergeUpper}.

  Otherwise $\ell>1$,
  and $Y$ and $Z$ are merged.\footnote{
      Note that in this case, if $\ell\ge 2$, $|Z| < |X|$,
       since (B-1) and (B-2) imply that
      $|X|\ge 2 |Y| \ge 2|Z| > |Z|$.
      This is the reason why Algorithm~\ref{alg:twoMergeSort} does not check
      for the condition $|X| < |Z|$ in lines
      \ref{algline:2mergefinalloopA}-\ref{algline:2mergefinalloopB} (unlike what
      is done on line~\ref{algline:2mergeXZ}).}
  We claim the resulting stack of
  runs satisfies case~(B),
  now with $\ell$ decremented by one. It is obvious that
  (B-1) and (B-3) still hold. (B-4) will still hold since by
  (B-3) and (B-4) the merge cost
  of the run formed by merging $Y$ and~$Z$ is at most
  $|Y| + |Z| + G_2(|Y|,m_Y) + H_2(|Z|,m_Z)$, and this is
  $\le H_2(|Y| +  |Z|,m_Y +  m_Z)$ by Lemma~\ref{lem:GHtwobounds}(a).
  To show (B-2) still holds, we must
  show that $|X| \ge |Y| +  |Z|$. To prove this,
  note that $\frac 1 2 |X| \ge |Y|$ by (B-1); thus from
  (B-2), also $\frac 1  2|X| \ge |Z|$.
  Hence $|X| \ge |Y| +  |Z|$.

  \smallskip

  \noindent
  {\bf (C): Encountered long run $Z$.}
  When case~(C) is first entered,
  the final run~$|Z|$ is long relative to~$|Y|$.
  The algorithm will repeatedly merge $X$ and~$Y$
  until $|Z|\le |X|$, at which point it merges $Y$ and~$Z$ and returns
  to case~(A). (The
  merge of $Y$ and $Z$ must eventually occur by the convention
  that $\QQ_0$ has infinite length.)
  The following conditions hold with $\ell\ge 2$:
  \begin{itemize}
    \setlength{\itemsep}{0pt}
    \item[(C-1)] $|\QQ_i| \ge 2 |\QQ_{i + 1}|$ for all $i < \ell - 2$.
      If $\ell \ge 3$, this includes $|W| \ge 2|X|$.
    \item[(C-2)] $|X| \ge |Y|$;
      i.e., $|\QQ_{\ell - 2}| \ge |\QQ_{\ell - 1}|$.
    \item[(C-3)] $|Y| < 2 |Z|$;
      i.e., $|\QQ_{\ell - 1}| < 2 |\QQ_\ell|$.
    \item[(C-4)] $w_{\QQ_i} \le G_2(|\QQ_i|, m_{\QQ_i})$ for all
      $i \le \ell - 2$.
    \item[(C-5)] $w_Y \le H_2(|Y|, m_Y)$.
    \item[(C-6)] $m_Z = 1$ and $w_Z = 0$,
      because $Z$ is an original run and has not undergone a merge.
  \end{itemize}
  By (C-3), the test on line~\ref{algline:2mergeXYZ} of Algorithm~\ref{alg:twoMergeSort}
  will now trigger a merge, either of $Y$ and~$Z$ or of $X$ and~$Y$,
  depending on the relative sizes of $X$ and~$Z$.
  We handle separately the cases $|Z| > |X|$ and $|Z| \le |X|$.
  \begin{itemize}
    \item Suppose $|Z|>|X|$. Then $\ell\ge 3$ and the algorithm
      merges $X$ and~$Y$. We claim that case~(C) still holds,
      now with $\ell$ decremented by~1. It is obvious that
      (C-1), (C-4) and (C-6) still hold.
      (C-5) will still hold, since by (C-4) and (C-5),
      the merge cost of the run obtained by
      merging $X$ and $Y$ is at most
      $|X| + |Y|+G_2(|X|,m_X) + H_2(|Y|,m_Y)$, and this is
      $\le H_2(|X| + |Y|,m_X +  m_Y)$ by Lemma~\ref{lem:GHtwobounds}(a)
      since $|X|\ge |Y|$.
      To see that (C-2) still holds, we argue exactly
      as in case~(B) to
      show that $|W| \ge |X| +  |Y|$. To prove this,
      note that $\frac 1 2 |W| \ge |X|$ by (C-1); thus from
      (C-2), $\frac 1 2 W \ge |Y|$. Hence $|W| \ge |X| +  |Y|$.
      To establish that (C-3) still holds, we must
      prove that $|X| +  |Y| <  2|Z|$.
      By the assumption that $|Z|>|X|$, this follows from
      $|Y| \le |X|$, which holds by (C-2).

    \item Otherwise, $|Z| \le |X|$ and $Y$ and $Z$ are merged. We claim that
      now case~(A) will hold.
      (A-1) will hold by (C-1).
      To show (A-2) will hold, we need $|Y| + |Z| \le 2|X|$:
      this holds by (C-2) and $|Z| \le |X|$.
      (A-3) will hold for $i < \ell - 1$ by (C-4).
      For $i=\ell - 1$, the merge cost $w_\YZ$ of the run obtained by
      merging $Y$ and~$Z$ is $\le H_2(|Y|,m_Y) + |Y| + |Z|$
      by (C-5) and (C-6). By (C-3) and Lemma~\ref{lem:GHtwobounds}(b)
      this is $\le G_2(|Y| + |Z|, m_Y + m_Z)$.
      Hence (A-3) will hold with $i = \ell - 1$.
  \end{itemize}

  That completes the proof of Theorem~\ref{thm:twomergeUpper}.
\end{proof}

Examination of the above proof shows why
the $2$-merge Algorithm~\ref{alg:twoMergeSort}
does not need to test
the condition $|X| < 2|Y|$ on line~\ref{algline:2mergeXYZ};
in contrast to what the $\alpha$-merge
Algorithm~\ref{alg:alphaMergeSort} does. In cases
(A) and~(B), the test will fail by conditions (A-1) and (B-1).
In case~(C), condition (C-3) gives
$|Y| < 2|Z|$, so an additional test would be redundant.

\subsection{Upper bound proof for \texorpdfstring{$\alpha$}{a}-merge sort}
\label{sec:alphaMergeUpperPf}

This section gives the proof Theorem~\ref{thm:alphamergeUpper}.
The general outline of the proof is similar to that of
Theorem~\ref{thm:twomergeUpper}; however, we must handle
a new, and fairly difficult, case (D). It is also necessary to
bound the total merge cost $\sum_i w_{\QQ_i}$ of all the runs
in the stack~$\mathcal \QQ$, instead of bounding
each individual merge cost $w_{\QQ_i}$.
We first prove a lemma stating properties of $G_\alpha$
and $H_\alpha$ which follow from
Lemmas \ref{lem:alphamerge-B}, \ref{lem:alphamerge-Ctwo}(b)
and~\ref{lem:alphamerge-CD}.  Parts (a) and~(b) of the lemma generalize
Lemma~\ref{lem:GHtwobounds}.

\begin{lemma}\label{lem:GHalphabounds}
  Suppose $n_1, n_2, m_1, m_2$ are positive integers, and
  $\varphi < \alpha < 2$.
  \begin{description}
    \item[\rm (a)] If $(\alpha  -  1) n_2 \le n_1$, then
      \[
        G_\alpha(n_1, m_1) + H_\alpha(n_2, m_2) + n_1 + n_2
        ~\le~
        H_\alpha(n_1  +   n_2, m_1  +   m_2).
      \]
    \item[\rm (b)] If $n_1 \le \frac{\alpha}{\alpha-1} \cdot n_2$, then
      \[
        H_\alpha(n_1,m_1) + n_1 + n_2 ~\le~ G_\alpha(n_1 +  n_2, m_1 +  m_2).
      \]
    \item[\rm (c)] If $n_1 \le \frac{\alpha}{\alpha-1} n_2$, then
      \[
          H_\alpha(n_1, m_1) ~\le~ G_\alpha(n_1,m_1) + G_\alpha(n_2, 1).
      \]
    \item[\rm (d)] If $k\le k_0(\alpha) +  1$ and
      $\frac{2^k (2\alpha - 1)}{\alpha - 1} n_3 \ge n_1 + n_2 + n_3$, then
      \[
       H_\alpha(n_1, m_1) + H_\alpha(n_2, m_2) + k \cdot n_3 + n_1 + 2n_2
       ~\le~ G_\alpha(n_1  +   n_2  +   n_3, m_1  +   m_2  +   1) .
      \]
  \end{description}
\end{lemma}

\begin{proof}
  Part (a) states that
  \[
    n_1 \cdot (d_\alpha - 1 + c_\alpha \log m_1)
     + n_2 \cdot (d_\alpha + c_\alpha \log m_2) + n_1 + n_2
    ~\le~
    (n_1  +   n_2) \cdot (d_\alpha + c_\alpha \log (m_2  +   m_2)).
  \]
  This is an immediate consequence of Lemma~\ref{lem:alphamerge-B}.

  Part~(b) states that
  \[
    n_1 \cdot (d_\alpha + c_\alpha \log m_1))
     + n_1 + n_2
    ~\le~
    (n_1 + n_2) \cdot (d_\alpha - 1 + c_\alpha \log (m_1  +   m_2)).
  \]
  This is exactly Lemma~\ref{lem:alphamerge-Ctwo}(b).

  The inequality of part~(c) states
  \[
    n_1\cdot(d_\alpha + c_\alpha \log m_1) ~\le~
       n_1\cdot(d_\alpha - 1 + c_\alpha \log m_1)
       + n_2 \cdot (d_\alpha - 1 + 0 ).
  \]
  After cancelling common terms, this is the same as
  $n_1 \le n_2 \cdot (d_\alpha - 1)$.  To establish this,
  it suffices to show that $d_\alpha - 1 \ge \frac{\alpha}{\alpha-1}$.
  Since $k_0(\alpha)\ge 1$,
  we have $d_\alpha\ge \frac{12(2\alpha-1)}{\alpha-1}+1$.
  And, since $\alpha>1$, we have $12(2\alpha-1)>\alpha$.
  Therefore $d_\alpha - 1 > \frac{\alpha}{\alpha-1}$, and
  (c)~is proved.

  Part~(d) states that
  \begin{eqnarray}
    \nonumber
    \lefteqn{ n_1 \cdot (d_\alpha + c_\alpha \log m_1) +
        n_2 \cdot (d_\alpha + c_\alpha \log m_2) +
        k \cdot n_3 + n_1 + 2 n_2} \\
    \label{eq:GHalphabounds_d}
    &\le& (n_1 + n_2 + n_3) \cdot (d_\alpha - 1 + c_\alpha \log (m_1  +   m_2  +   1)).
    \hspace*{1in}
  \end{eqnarray}
  Lemma~\ref{lem:alphamerge-CD} implies that
  \begin{eqnarray*}
  \lefteqn{ n_1 \cdot (d_\alpha + c_\alpha \log m_1) +
      n_2 \cdot (d_\alpha + c_\alpha \log m_2) +
      k \cdot n_3 + n_1 + 2 n_2} \\
  &\le&
      n_1 \cdot (d_\alpha - 1 + c_\alpha \log m_1) +
      (n_2 + n_3) \cdot (d_\alpha -1 + c_\alpha \log (m_2  +   1)) .
  \end{eqnarray*}
  The desired inequality (\ref{eq:GHalphabounds_d}) follows easily.
\end{proof}

We now prove Theorem~\ref{thm:alphamergeUpper}.

\begin{proof}[Proof of Theorem~\ref{thm:alphamergeUpper}]
  We describe the $\alpha$-merge algorithm using four invariants
  (A), (B), (C), (D) for the stack; and analyze what action is taken
  in each situation. Initially, the stack contains a single original $\QQ_1$, so
  $\ell = 1$ and $m_{\QQ_1} = 1$ and $w_{\QQ_1}=0$, and case~(A) applies.

  \smallskip

  \noindent
  {\bf (A): Normal mode.} The stack satisfies
  \begin{itemize}
    \setlength{\itemsep}{0pt}
    \item[(A-1)]
      $|\QQ_i| \ge \alpha |\QQ_{i + 1}|$ for all $i < \ell  -  1$. This includes
      $|X| \ge \alpha|Y|$ if $\ell \ge 2$.
    \item[(A-2)] $\alpha |Y| \ge |Z|$; i.e.\ $\alpha |\QQ_{\ell - 1}| \ge |\QQ_\ell|$.
    \item[(A-3)] $\sum_{i = 1}^\ell w_{\QQ_i} ~\le~
      \sum_{i = 1}^\ell G_\alpha(|\QQ_i|, m_{\QQ_i})$.
  \end{itemize}

  If $\ell \ge 2$, (A-1) and (A-2) imply $|X|\ge |Z|$, i.e.\
  $|\QQ_{\ell - 2}| \ge |\QQ_\ell|$.
  The $\alpha$-merge algorithm does one of the following:
  \begin{itemize}
    \item If $\alpha |Z| \le |Y|$ and there are no more original runs to load,
      then it goes to case~(B). Condition (B-1) holds by (A-1). (B-3) holds by
      (A-3) since
      $G_\alpha(|\QQ_\ell|, m_{\QQ_\ell}) \le H_\alpha(|\QQ_\ell|, m_{\QQ_\ell})$.
      Condition (B-2) states that $(\alpha - 1) |Z| \le |Y|$ and this holds since
      $\alpha |Z| \le |Y|$.

    \item If $\alpha |Z| \le |Y|$ and there is another original run to load, then
      the algorithm loads the next run as $\QQ_{\ell + 1}$.
      \begin{itemize}
        \item[$\circ$] If $|\QQ_{\ell + 1}| \le \alpha |\QQ_\ell|$,
          then we claim that case~(A) still holds after $\ell$ is incremented
          by one. In particular, (A-1) and $\alpha |Z| \le |Y|$ imply that
          (A-1) will still hold since $\alpha |Z| \le |Y|$ is the same as
          $\alpha |\QQ_\ell| \le |\QQ_{\ell - 1}|$. (A-2) still holds by the
          assumed bound on $|\QQ_{\ell + 1}|$. Condition (A-3) still holds since
          $|\QQ_{\ell + 1}|$ is an original run so $m_{\QQ_{\ell + 1}}=1$ and
          $w_{\QQ_{\ell + 1}} = 0$.
        \item[$\circ$] Otherwise $|\QQ_{\ell + 1}| > \alpha |\QQ_\ell|$, and
          we claim that case~(C) below holds with
          $\ell$ incremented by one.
          (C-1) holds by (A-1).
          For (C-2), we need $|\QQ_{\ell - 1}| \ge (\alpha  -  1)|\QQ_\ell|$, i.e.
          $|Y| \ge (\alpha  -  1)|Z|$: this follows trivially from
          $|Y| \ge \alpha |Z|$. (C-4) holds by (A-3), since
          $G_\alpha(|\QQ_\ell|,m_{\QQ_\ell}) \le H_\alpha(|\QQ_\ell|,m_{\QQ_\ell})$ and
          since $\QQ_{\ell + 1}$ is an original run so $m_{\QQ_{\ell + 1}} = 1$ and
          $w_{\QQ_{\ell + 1}} = 0$. To have (C-3) hold, we need
          $|Z| \le \frac{\alpha}{(\alpha - 1)}|\QQ_{\ell + 1}|$. This follows
          from the hypothesis $|\QQ_{\ell + 1}| > \alpha |Z|$ and $\alpha > 1$.
          Finally, (C-5) will hold since $\QQ_{\ell + 1}$ is an original run.
      \end{itemize}

    \item If $\alpha |Z| > |Y|$, then $\ell\ge 2$ . In this case, since
      $|Z|\le|X|$, the algorithm merges the two runs $Y$ and~$Z$. We claim the
      resulting stack satisfies case~(A) with $\ell$ decremented by one.
      It is obvious that (A-1) still holds. (A-3) still holds by
      Lemma~\ref{lem:Gbound}. For (A-2), we need $|Y| + |Z| \le \alpha|X|$;
      this follows from $|Y| \le \frac{1}{\alpha} |X|$ and $|Z| \le |X|$ and
      $1 + \frac{1}{\alpha} < \alpha$ as $\varphi<\alpha$.
  \end{itemize}

  \noindent
  {\bf (B): Wrapup mode, lines 15-18 of
  Algorithm~\ref{alg:alphaMergeSort}.}
  There are no more original runs to process. The
  entire input has been combined into the
  runs $\QQ_1, \ldots, \QQ_\ell$ and they satisfy:
  \begin{itemize}
  \item[(B-1)] $|\QQ_i| \ge \alpha |\QQ_{i + 1}|$ for all $i < \ell  -  1$.
     This includes $|X| \ge \alpha |Y|$ if $\ell\ge 2$.
  \item[(B-2)] $|Y| \ge (\alpha  -  1)|Z|$; i.e.,
      $|\QQ_{\ell - 1}| \ge (\alpha  -  1)|\QQ_\ell|$.
  \item[(B-3)]
      $\sum_{i = 1}^\ell w_{\QQ_i} ~\le~
           \sum_{i = 1}^{\ell - 1}  G_\alpha(|\QQ_i|, m_{\QQ_i}) +
           H_\alpha(|\QQ_\ell|, m_{\QQ_\ell})$.
  \end{itemize}
  If $\ell = 1$, the run $Z = \QQ_1$ contains the entire input in sorted order
  and the algorithm terminates. The total merge cost is
  $\le H_\alpha(|Z|, m_Z)$. This is $n \cdot (d_\alpha + c_\alpha \log m )$
  as needed for Theorem~\ref{thm:alphamergeUpper}.

  If $\ell > 1$, then $Y$ and $Z$ are merged.\footnote{
      Note that in this case, if $\ell\ge 2$, $|Z| < |X|$ since (B-1) and (B-2) imply that
      $|X|\ge \alpha |Y| \ge \alpha(\alpha - 1)|Z|$ and $\alpha^2-\alpha > 1$
      since $\varphi < \alpha$.
      This is the reason why Algorithm~\ref{alg:alphaMergeSort} does not check
      for the condition $|X| < |Z|$ in lines
      \ref{algline:amergefinalloopA}-\ref{algline:amergefinalloopB} (unlike what
      is done on line~\ref{algline:amergeXZ}).
  }
  We claim that the resulting stack of runs satisfies (B), now with $\ell$ decremented by one.
  It is obvious that (B-1) will still holds.
  (B-3) will still hold since merging $Y$ and~$Z$
  adds $|Y| + |Z|$ to the total merge cost and since
  \[
    G(|Y|, m_Y) + H(|Z|, m_Z) + |Y| + |Z|
    ~\le~
    H(|Y|  +   |Z|, m_Y  +   m_Z)
  \]
  by Lemma~\ref{lem:GHalphabounds}(a) since $|Y| \ge (\alpha - 1)|Z|$ by (B-2).
  To show (B-2) will still hold, we must show that
  $|X| \ge (\alpha  -  1)(|Y|  +   |Z|)$. To prove this, note that
  $\frac{1}{\alpha} |X| \ge |Y|$ by (B-1); thus from (B-2),
  $\frac{1}{\alpha(\alpha-1)}|X| \ge |Z|$. This gives
  $\bigl( \frac 1 \alpha + \frac 1 {\alpha(\alpha-1)} \bigr)|X| \ge |Y|  +   |Z|$;
  hence $|X| \ge (\alpha  -  1)(|Y|  +   |Z|)$.

  \smallskip

  \noindent
  {\bf (C): Encountered long run $Z$.}
  When case~(C) is first entered, the final run~$Z$ is long relative to~$Y$.
  The algorithm will repeatedly merge $X$ and~$Y$ as long as $|Z| < |X|$, staying
  in case~(C). Once $|Z|\le|X|$, as discussed below,
  there are several possibilities. First, it may be
  that case~(A) already applies. Otherwise, $Y$ and~$Z$ are merged, and the algorithm
  proceeds to either case~(A) or case~(D).

  Formally, the following conditions hold during case~(C)
  with $\ell\ge 2$:
  \begin{itemize}
    \setlength{\itemsep}{0pt}
    \item[(C-1)] $|\QQ_i| \ge \alpha |\QQ_{i + 1}|$ for all
      $i < \ell  -  2$. If $\ell\ge 4$, this includes $|W| \ge \alpha|X|$.
    \item[(C-2)] $|X| \ge (\alpha  -  1) |Y|$;
      i.e., $|\QQ_{\ell - 2}| \ge (\alpha  -  1)|\QQ_{\ell - 1}|$.
    \item[(C-3)] $|Y| \le \frac \alpha{(\alpha - 1)} |Z|$;
      i.e., $|\QQ_{\ell  -  1}| \le \frac{\alpha}{(\alpha - 1)} |\QQ_\ell|$.
    \item[(C-4)] $\sum_{i = 1}^\ell w_{\QQ_i} \le
      \sum_{i = 1}^{\ell - 2} G_\alpha(|\QQ_i|, m_{\QQ_i}) + H_\alpha(|Y|, m_Y)$.
    \item[(C-5)] $Z$ is an original run, so $m_Z=1$ and $w_Z = 0$.
  \end{itemize}

  It is possible that no merge is needed, namely if $|X| \ge \alpha |Y|$ and
  $|Y| \ge \alpha |Z|$. In this case we claim that case (A) holds. Indeed,
  (A-1) will hold by (C-1) and since $|X| \ge \alpha |Y|$. Condition (A-2)
  holds by $|Y| \ge \alpha |Z|$. Condition (A-3) follows from (C-4) and the
  fact that, using (C-3), Lemma~\ref{lem:GHalphabounds}(c) gives the inequality
  $H_\alpha(|Y|, m_Y) \le G_\alpha(|Y|, m_Y) +  G_\alpha(|Z|, 1)$.\footnote{
    It is this step which requires us to bound the total merge cost $\sum_i
    w_{\QQ_i}$ instead of the individual merge costs $w_{\QQ_i}$. Specifically,
    $G_\alpha(|Y|, m_Y)$ may not be an upper bound for $w_Y$.
  }

  Otherwise, a merge occurs. The cases $|Z| > |X|$ and $|Z|\le|X|$
  are handled separately:

  \begin{itemize}
    \item Suppose $|Z| > |X|$. We have $\ell\ge 3$ by the convention that
      $|\QQ_0| = \infty$, and the algorithm
      merges $X$ and~$Y$. We claim that case~(C) still holds, now
      with $\ell$ decremented by~1.
      It is obvious that (C-1) and (C-5) will still hold.
      (C-4) will still hold since merging $X$ and~$Y$
      adds $|X| + |Y|$ to the total merge cost, and since
      \[
        G_\alpha(|X|, m_X) + H_\alpha(|Y|, m_Y) + |X| + |Y|
          ~\le~ H_\alpha(|X| + |Y|, m_X + m_Y)
      \]
      by Lemma~\ref{lem:GHalphabounds}(a) since $|X| \ge (\alpha - 1) |Y|$
      by (C-3).

      To see that (C-2) still holds, we argue exactly as in case~(B): We must
      show that $|W| \ge (\alpha - 1)(|X| + |Y|)$. To prove this, note that $\frac{1}{\alpha} |W| \ge |X|$ by (C-1); thus from (C-2),
      $\frac{1}{\alpha(\alpha-1)} W \ge |Y|$. This gives
      $\left(
          \frac{1}{\alpha} + \frac{1}{\alpha(\alpha - 1)}
      \right) |W| \ge |X| + |Y|$;
      hence $|W| \ge (\alpha  -  1)(|X|  +   |Y|)$.

      To establish that (C-3) still holds, we must
      prove that $|X| + |Y| \le \frac{\alpha}{\alpha - 1}|Z|$. Since
      $|Z| > |X|$, it suffices to show $|X| + |Y| \le
      \frac{\alpha}{\alpha - 1}|X|$, or equivalently
      $|Y| \le \frac{1}{\alpha - 1}|X|$. This holds by (C-2).

    \item Otherwise $|Z| \le |X|$, so $\ell \ge 2$,
      and $Y$ and $Z$ are merged. The analysis
      splits into two cases, depending on whether $|Y| + |Z| \le \alpha |X|$
      holds.

      First, suppose $|Y| + |Z| \le \alpha |X|$. Then we
      claim that, after the merge of $Y$ and~$Z$, case~(A) holds.
      Indeed, (A-1) will hold by (C-1). (A-2) will hold by
      $|Y| + |Z| \le \alpha |X|$. Condition (A-3) will hold by (C-4) and (C-5)
      and the fact that, using (C-3), Lemma~\ref{lem:GHalphabounds}(b) gives
      the inequality
      $H_\alpha(|Y|, m_Y) + |Y| + |Z|\le G_\alpha(|Y| +  |Z|, m_Y +  1)$.

      Second, suppose $|Y| + |Z| > \alpha |X|$ and thus $\ell \ge 3$.
      We claim that, after the merge of $Y$ and~$Z$, case (D) holds.
      (D-1) holds by (C-1).
      (D-2) holds with $k = 0$ so $Z_1$ is the empty run, and
      with $Z_2$ and~$Z_3$ equal to the just-merged runs $Y$ and~$Z$
      (respectively).
      (D-3) holds since $k=0$ and $|X| \ge |Z|$.
      (D-4) holds by (C-2) and the choice of $Z_1$ and~$Z_2$.
      (D-5) holds by (C-3).
      (D-6) holds since is the same as the assumption
      that $|Y| + |Z| > \alpha |X|$.
      (D-7) holds since, by (C-3),
          $|Y|  +   |Z| \le
          \Bigl(\frac{\alpha}{\alpha - 1}  +   1\Bigr) |Z| =
          \frac{2\alpha - 1}{\alpha - 1} |Z|$.
      Finally, we claim that (D-8) holds by (C-4).
      To see this, note
      that with $k = 0$, the quantity ``$(k  +   1) Z_3 + Z_2$'' of (D-8) is
      equal to the cost $|Y| + |Z|$ of merging $Y$ and~$Z$,
      and that the quantities ``$m_Z  -  1$'' and ``$|Z_1| + |Z_2|$''
      of (D-8) are the same as our $m_Y$ and $|Y|$.
  \end{itemize}

  \noindent
  {\bf (D): Wrapping up handling a long run.}
  In case~(D), the original long run,
  which was earlier called ``$Z$'' during case~(C), is now
  called ``$Z_3$'' and has been merged with runs at the
  top of the stack to form the current top stack element $Z$.
  This $Z$ is equal to the merge of three runs $Z_1$, $Z_2$ and~$Z_3$.
  The runs $Z_2$ and~$Z_3$ are the two runs $Y$ and~$Z$ which
  were merged when leaving case~(C) to enter case~(D). The run
  $Z_1$ is equal to the merge of $k$ many runs $U_k, \ldots,U_1$
  which were just below the top of the stack when leaving case~(C)
  to enter case~(D). Initially $k=0$, so $Z_1$ is empty.

  The runs $Z_2$ and~$Z_3$ do not change while the algorithm is
  in case~(D).
  Since $Z_3$ was an original run, $m_{Z_3}=1$. In other words,
  $m_Z - 1 = m_{Z_1} +  m_{Z_2}$.

  There are two possibilities for how the merge algorithm proceeds in
  case~(D). In the simpler case, it merges Y and Z, and either
  goes to case~(A) or stays in case~(D).
  If it stays in case~(D), $k$~is incremented by~1.
  In the more complicated case, it merges X and Y, then merges the
  resulting run with Z, and then again either goes to (A) or stays in (D).
  In this case, if it stays in~(D), $k$ is incremented by~2.
  Thus $k$ is equal to the number of merges that have been performed.
  We will show that $k<k_0(\alpha)$ must always hold.

  Formally, there is an integer $k < k_0(\alpha)$
  and there exists runs $Z_1$, $Z_2$ and
  $Z_3$ (possibly $Z_1$ is empty) such that $\ell\ge 2$ and the following hold:
  \begin{itemize}
      \setlength{\itemsep}{0pt}
      \item[(D-1)] $|\QQ_i| \ge \alpha |\QQ_{i + 1}|$ for all
          $i < \ell  -  1$. This includes $X\ge \alpha Y$.
      \item[(D-2)] $Z$ is equal to the merge of three runs $Z_1,Z_2,Z_3$.
      \item[(D-3)] $|Y| \ge \alpha^k |Z_3|$.
      \item[(D-4)] $|Y|\ge (\alpha - 1)(|Z_1| + |Z_2|)$.
      \item[(D-5)] $|Z_3| \ge \frac{\alpha - 1}{\alpha} |Z_2|$.
      \item[(D-6)] $\alpha |Y| < |Z|$.
      \item[(D-7)] $|Z| \le  \frac{2^k (2\alpha - 1)}{\alpha - 1} |Z_3|$.
      \item[(D-8)] $\sum_{i = 1}^\ell w_{\QQ_i} ~\le~
          \sum_{i = 1}^{\ell - 1} G_\alpha(|\QQ_i|, m_{\QQ_i}) + (k  +   1) |Z_3| + |Z_2| +
          H_\alpha(|Z_1| +  |Z_2|, m_Z  -  1)$.
  \end{itemize}
  We claim that conditions (D-3), (D-4) and (D-6) imply that $k<k_0(\alpha)$.
  To prove this, suppose $k\ge k_0(\alpha)$. From the definition
  of $k_0$, this implies that $\alpha \ge \frac 1 {\alpha^k} + \frac 1{\alpha-1}$.
  (D-3) gives $\frac1{\alpha^k} |Y| \ge |Z_3|$;
  (D-4) gives $\frac1{\alpha-1} |Y| \ge |Z_1| + |Z_2|$.
  With (D-6) and $|Z| = |Z_1| +  |Z_2| +  |Z_3|$, these imply
  \[
  |Z| ~>~ \alpha|Y|
       ~\ge~ \Bigr(\frac1{\alpha^k} + \frac1{\alpha-1}\Bigl ) |Y|
       ~\ge~ |Z_1| + |Z_2| + |Z_3| ~=~ |Z| ,
  \]
  which is a contradiction.

  By (D-6) and the test in line~\ref{algline:amergeXYZ} of
  Algorithm~\ref{alg:alphaMergeSort}, the algorithm must perform a merge,
  either of $X$ and~$Y$ or of $Y$ and~$Z$, depending on the relative sizes of
  $X$ and~$Z$. The cases of $|Z| \le |X|$ and $|Z| > |X|$ are handled
  separately.
  \begin{itemize}
    \item Suppose $|Z| \le |X|$.  Therefore, the algorithm merges $Y$ and~$Z$.

      In addition, suppose that $\alpha |X| \ge |Y| + |Z|$. We claim that this
      implies case~(A) holds after the merge. Indeed, (D-1) implies that
      condition (A-1) holds. The assumption $\alpha |X| \ge |Y| + |Z|$ gives
      that (A-2) will hold. For (A-3), we argue by applying
      Lemma~\ref{lem:GHalphabounds}(d) with $n_1 = |Y|$ and $n_2 = |Z_1| +
      |Z_2|$ and $n_3 = |Z_3|$ and with $k +  2$ in place of~$k$.
      For this, we need $k +  2 \le k_0(\alpha) +  1$ as was already proved and
      also need
      \[
        \frac{2^{k+2} (2\alpha  -  1)}{\alpha - 1} |Z_3|
        ~\ge~ |Y|+ |Z|.
      \]
      This is true by (D-7) since $|Y| + |Z| < (1 + 1/\alpha)|Z|$ by (D-6)
      and since $1 + 1 / \alpha < 4$. We also have $m_Z  -  1 = m_{Z_1}  +   m_{Z_2} > 0$.
      Thus, Lemma~\ref{lem:GHalphabounds}(d) implies
      \[
        H_\alpha(|Y|, m_Y) + H_\alpha(|Z_1| + |Z_2|, m_Z - 1) +
        (k + 2) \cdot |Z_3| + |Y| + 2(|Z_1| + |Z_2|)
        ~\le~
        G_\alpha(|Y|  +   |Z|, m_Y +  m_Z).
      \]
      Since $\alpha < 2$, $G_\alpha(|Y|, m_Y) \le H_\alpha(|Y|, m_Y)$.
      Therefore, since $|Z| = |Z_1|  +   |Z_2|  +   |Z_3|$,
      \[
        G_\alpha(|Y|,m_Y) + (k  +   1) |Z_3| + |Z_2| +
            H_\alpha(|Z_1|  +   |Z_2|, m_Z  -  1) + |Y| +|Z|
        ~\le~
        G_\alpha(|Y| + |Z|, m_Y + m_Z).
      \]
      Since the cost of merging $Y$ and~$Z$ is equal to $|Y| + |Z|$, this
      inequality plus the bound (D-8) implies that (A-3) will hold after the
      merge.

      Alternately, suppose $\alpha|X| < |Y| + |Z|$. In this case,
      $Y$ and $Z$ are merged:
      $k$~will be incremented by~1 and $Y$~will become part of~$Z_1$.
       We claim that case~(D) still holds after the merge.
      Clearly (D-1) still holds.
      (D-2) still holds with $Z_1$ now including~$Y$. (That is,
      $Y$ becomes~$U_{k+1}$.) (D-1) and (D-3) imply
      $|X| \ge \alpha |Y| \ge \alpha^{k + 1} |Z_3|$,
      so (D-3) will still hold.
      (D-4) implies
      $\frac{|Y|}{\alpha - 1} \ge |Z_1| + |Z_2|$.
      Therefore, (D-1) gives
      \[
        \frac{|X|}{\alpha - 1} ~\ge~ \frac{\alpha |Y|}{\alpha - 1}
          ~=~ |Y| + \frac{|Y|}{\alpha - 1} ~\ge~ |Y| + |Z_1| + |Z_2|.
      \]
      Thus (D-4) will still hold after the merge (since, after the merge,
      $Y$ becomes part of~$Z_1$).
      The hypothesis $\alpha|X| < |Y| + |Z|$ implies (D-6) will still hold.
      By (D-6) we have $|Z| > \alpha |Y| > |Y|$, so (D-7) gives
      \[
      |Z| + |Y| ~\le~ 2 |Z| ~\le~ \frac{2^{k + 1} (2\alpha  -  1)}{\alpha - 1} |Z_3|.
      \]
      Thus (D-7) will still hold. Finally,
      by (D-4), we may apply Lemma~\ref{lem:GHalphabounds}(a)
      with $n_1=|Y|$ and $n_2=|Z_1| +  |Z_2|$ to obtain
      \[
        G_\alpha(|Y|,m_Y) + H_\alpha(|Z_1| +  |Z_2|,m_Z - 1)
         + |Y| + |Z_1| + |Z_2| + |Z_3|
         ~\le~ H_\alpha(|Y| + |Z_1| +  |Z_2|,m_Y+m_Z - 1) + |Z_3|.
      \]
      Since the cost of merging $Y$ and $Z$ is
      $|Y| + |Z_1| + |Z_2| + |Z_3|$ and since $k$ is incremented
      by~1 after the merge, this implies that (D-8) will still
      hold after the merge.

    \item Now suppose $|Z| > |X|$, so $\ell\ge 3$.
      In this case, algorithm merges $X$ and $Y$; the result
      becomes the second run on the stack, which we denote $(XY)$.
      We claim that immediately after this, the
      algorithm merges the combination $(XY)$
      of $X$ and~$Y$ with the run $Z$.
      Indeed, since $\varphi<\alpha$, we have $\alpha > 1+ 1/\alpha$
      and therefore by (D-6) and by the assumed bound on $|X|$
      \[
          \alpha |Z| ~>~ \frac{1}{\alpha} |Z| + |Z| ~>~ |Y| + |X|.
      \]
      Thus, the test on line~\ref{algline:amergeXYZ} of
      Algorithm~\ref{alg:alphaMergeSort} triggers a second
      merge operation.  Furthermore, since $\varphi<\alpha$,
      we have $1 > \frac{1}{\alpha^2} + \frac{1}{\alpha^2 (\alpha - 1)}$
      and thus, using (D-1), (D-3) and (D-4),
      \[
          |W| ~>~ \frac{1}{\alpha^{k + 2}} |W| +
                  \frac{1}{\alpha^2 (\alpha  -  1)} |W|
              ~\ge~ \frac{1}{\alpha^k} |Y| +
                  \frac{1}{\alpha - 1} |Y|
              ~\ge~ |Z_1| + |Z_2| + |Z_3|= |Z|.
      \]
      With $|W| > |Z|$, the second merge acts to merge $(XY)$ and~$Z$
      instead of $W$ and $(XY)$. After the second merge,
      the top three runs $X,Y,Z$ on the stack
      have been merged, with an additional merge cost
      of $2|X| +  2|Y| +  |Z|$.  As we argue next, the algorithm
      now either transitions to case~(A)
      or stays in case~(D), depending whether
      $\alpha |W| \ge |X| + |Y| + |Z|$ holds.

      First, suppose that $\alpha |W| \ge |X| + |Y| + |Z|$.
      We claim that this implies case~(A) holds after the two merges.
      Indeed, (D-1) implies that condition (A-1) will hold.
      The assumption $\alpha |W| \ge |X| + |Y| + |Z|$ gives
      that (A-2) will hold.
      For (A-3), we argue as follows. We apply
      Lemma~\ref{lem:GHalphabounds}(d) with $n_1 = |X|  +   |Y|$
      and $n_2 = |Z_1|  +   |Z_2|$ and $n_3 = |Z_3|$
      and with $k+2$ in place of~$k$.
      For this, we need $k +  2 \le k_0(\alpha)+1$ as was already proved
      and also need
      \begin{equation}\label{eq:ZthreeXYZ}
          \frac{2^{k + 2} (2\alpha  -  1)}{\alpha - 1} |Z_3|
          ~\ge~ |X| + |Y| + |Z|.
      \end{equation}
      To prove~(\ref{eq:ZthreeXYZ}), first note that
      that $|X| + |Y| +|Z| < (2 + 1 / \alpha)|Z|$ by
      (D-6) and the assumption that $|Z| > |X|$;
      then (\ref{eq:ZthreeXYZ}) follows from (D-7)
      and the fact that $2+1/\alpha<4$.
      We have $m_Z  -  1 = m_{Z_1}  +   m_{Z_2} > 0$.
      Thus, Lemma~\ref{lem:GHalphabounds}(d) implies
      \begin{eqnarray}
        \nonumber
        \lefteqn{H_\alpha(|X|  +   |Y|, m_X  +   m_Y) + H_\alpha(|Z_1|  +   |Z_2|, m_Z  -  1) +
                          (k  +   2) \cdot |Z_3|
                        + |X| + |Y| + 2(|Z_1| + |Z_2|)} \\
        \label{eq:Dtwo1}
        &\le& G_\alpha(|X|  +   |Y|  +   |Z|, m_X  +   m_Y  +   m_Z).
              \hspace*{3.3in}
      \end{eqnarray}
      By (D-1), we have $(\alpha  -  1) |Y| < \alpha |Y| \le |X|$.
      Thus Lemma~\ref{lem:GHalphabounds}(a), with $n_1 = |X|$ and
      $n_2 = |Y|$ and $m_1=m_X$ and $m_2 = m_Y$ gives
      \begin{equation}\label{eq:Dtwo2}
        G_\alpha(|X|, m_X) + H_\alpha(|Y|, m_Y) + |X| + |Y|
         ~\le~ H_\alpha(|X|  +   |Y|, m_X  +   m_Y).
      \end{equation}
      Since $\alpha<2$, $G_\alpha(|Y|, m_Y) < H_\alpha(|Y|, m_Y)$.
      So, using $|Z| = |Z_1| +  |Z_2| +  |Z_3|$,
      (\ref{eq:Dtwo1}) and~(\ref{eq:Dtwo2}) imply
      \begin{eqnarray*}
        \lefteqn{G_\alpha(|X|, m_X) +
              G_\alpha(|Y|, m_Y) +
              (k  +   1) |Z_3| + |Z_2| +
              H_\alpha(|Z_1|  +   |Z_2|, m_Z  -  1) + 2|X|  +   2|Y|  +  |Z|} \\
        &\le& G_\alpha(|X|  +   |Y|  +   |Z|, m_X  +   m_Y  +   m_Z).
                \hspace*{3.3in}
      \end{eqnarray*}
      Since the cost of the two merges combining $X$, $Y$, and~$Z$
      was $2|X|  +   2|Y|  +  |Z|$,
      the last inequality and (D-8) imply
      that (A-3) will hold after the two merges.

      Alternately, suppose $\alpha |W| < |X| + |Y| + |Z|$. This implies
      $\ell \ge 4$. We claim that in this
      case~(D) still holds after the two merges combining
      $X$, $Y$ and~$Z$, and with $k$ incremented by~2.
      Certainly, (D-1) still holds.
      (D-2) is also still true:
      $Z_2$ and~$Z_3$ are unchanged, and $Z_1$ will
      include $X$ and~$Y$ (as $U_{k+2}$ and $U_{k+1}$).
      (D-1) and (D-3) imply
      \[
        |W| ~\ge~ \alpha |X| ~\ge~ \alpha^2 |Y| ~\ge~ \alpha^{k + 2} |Z_3|,
      \]
      so (D-3) will still hold.
      (D-4) implies $\frac{|Y|}{\alpha - 1} \ge |Z_1|  +   |Z_2|$.
      Hence, (D-1) and (D-4) give
      \begin{eqnarray*}
        \frac{|W|}{\alpha - 1}
          &\ge& \frac{\alpha |X|}{\alpha - 1}
          ~=~ \frac{|X|}{\alpha - 1} + |X|
          ~\ge~ \frac{\alpha |Y|}{\alpha - 1} + |X|
          ~=~ |X| + |Y| + \frac{|Y|}{\alpha - 1} \\
        &\ge& (|X| + |Y| + |Z_1|) + |Z_2|;
      \end{eqnarray*}
      i.e., (D-4) will still hold.
      (D-5) is obviously still true.
      $\alpha|W| < |X|  +   |Y|  +   |Z|$ implies (D-6)
      will still hold.
      The already proved equation~(\ref{eq:ZthreeXYZ}) implies
      that (D-7) will still hold.
      Finally, we need to establish (D-8).

      By (D-4), $|Y| \ge (\alpha  -  1) (|Z_1|  +   |Z_2|)$;
      and thus by (D-1) and (D-3),
      $|X| \ge \alpha|Y| \ge (\alpha - 1)(|Y| +  |Z_1| +  |Z_2|)$
      Therefore, we can apply Lemma~\ref{lem:GHalphabounds}(a) twice,
      first with $n_1 = |Y|$ and $n_2 = |Z_1| +  |Z_2|$
      and then with $n_1 = |X|$ and $n_3 = |Y| +  |Z_1| +  |Z_2|$,
      to obtain
      \begin{eqnarray*}
      \lefteqn{G_\alpha(|X|, m_X) + G_\alpha(|Y|, m_y) +
                 H_\alpha(|Z_1|  +   |Z_2|, m_Z  -  1) + 2|X| + 2|Y| + |Z|} \\
      &<& G_\alpha(|X|, m_X) + G_\alpha(|Y|, m_y) +
                      H_\alpha(|Z_1|  +   |Z_2|, m_Z  -  1) + |X| + 2|Y| + 2|Z| \\
      &\le& G_\alpha(|X|, m_X) + H_\alpha(|Y|  +   |Z_1|  +   |Z_2|, m_Y  +   m_Z  -  1) +
                  |X| + |Y| + |Z| + |Z_3| \\
      &\le& H_\alpha(|X|  +   |Y|  +   |Z_1|  +  |Z_2|, m_X  +   m_Y  +   m_Z  -  1) + 2|Z_3|
      \end{eqnarray*}
      where the first inequality uses the assumption $|X| < |Z|$
      and the other two inequalities use
      Lemma~\ref{lem:GHalphabounds}(a) and $|Z| = |Z_1| +  |Z_2| +  |Z_3|$.
      From this, it is easily seen that (D-8) will still hold after the
      two merges combining $X$, $Y$ and~$Z$: this is because the
      additional merge cost is $2|X|  +   2|Y|  +   |Z|$ and
      since $k$ will be incremented by~2.
  \end{itemize}
  This completes the proof of Theorem~\ref{thm:alphamergeUpper}.
\end{proof}

\section{Experimental results}\label{sec:experiments}
This section reports some computer experiments comparing the $\alpha$-stack
sorts, the $\alpha$-merge sorts, Timsort, Shivers sort, adaptive Shivers sort
and powersort. The test
sequences use the following model. We only measure merge costs, so the inputs to
the sorts are sequences of run lengths (not arrays to be sorted).  Let $\mu$ be
a distribution over integers. A sequence of $m$ run lengths is chosen by
choosing each of the $m$ lengths independently according to the
distribution~$\mu$. We consider two types of distributions for~$\mu$:
\begin{enumerate}
  \item The uniform distribution over numbers between $1$ and $100$,
  \item A mixture of the uniform distribution over integers between $1$ and
    $100$ and the uniform distribution over integers between $10000$ and
    $100000$, with mixture weights $0.95$ and $0.05$. This distribution was
    specially tailored to work better with $3$-aware algorithms while still
    being formulated in a general way that avoids favoring any
    particular algorithm.
\end{enumerate}
We also experimented with power law distributions.  However,
they gave very similar
results to the uniform distributions so we do not report these results
here.

In Figures
\ref{figure:alpha-stack-sort-plot}
and~Figure~\ref{figure:alpha-merge-sort-plot}, we estimate the ``best'' $\alpha$
values for the $\alpha$-merge and $\alpha$-stack sorts under the uniform
distribution. The experiments show that the best value for $\alpha$ for both
types of algorithms is around the golden ratio, or even slightly lower.  For $\alpha$
at or below $\varphi$, the results start to show oscillation,
albeit within a small range.  We discuss
oscillation for other sorts below, but we do not know why the results $\alpha$-merge sort
oscillate for small values of~$\alpha$.

\begin{figure}
    \centering
    \input{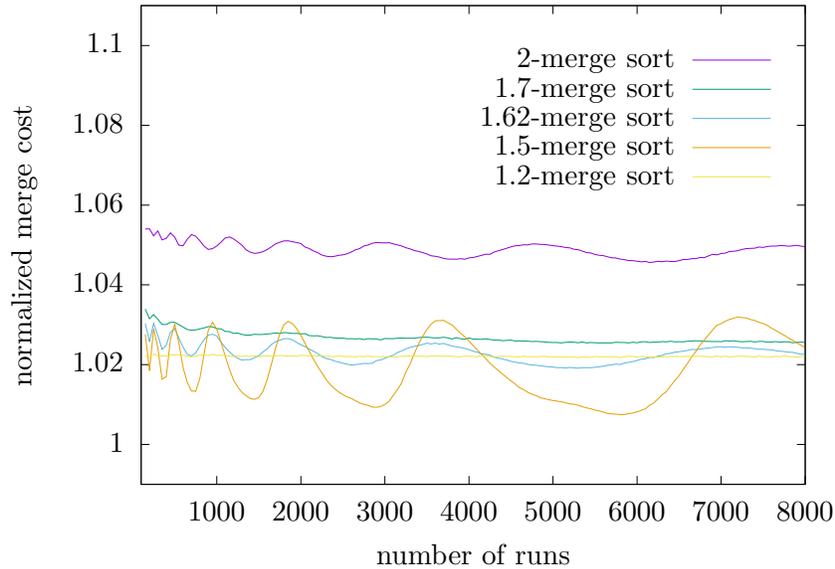}
    \caption{Comparison between $\alpha$-merge sorts for different $\alpha$ on
    uniform distribution over integers between 1 and 100.
    For all our figures, the $x$-axis
    shows the ``number
    of runs'', namely the number of presorted subsequences (ascending or
    descending) in the input sequence.  The $y$-axis shows the ``normalized
    merge cost'', namely the ratio of the total merge cost and $n\log m$.
    The data points reported in
    Figures \ref{figure:alpha-stack-sort-plot}-\ref{figure:allsortsMixed},
    are the average of 100 random trials.}
    \label{figure:alpha-stack-sort-plot}
\end{figure}

\begin{figure}
    \centering
    \input{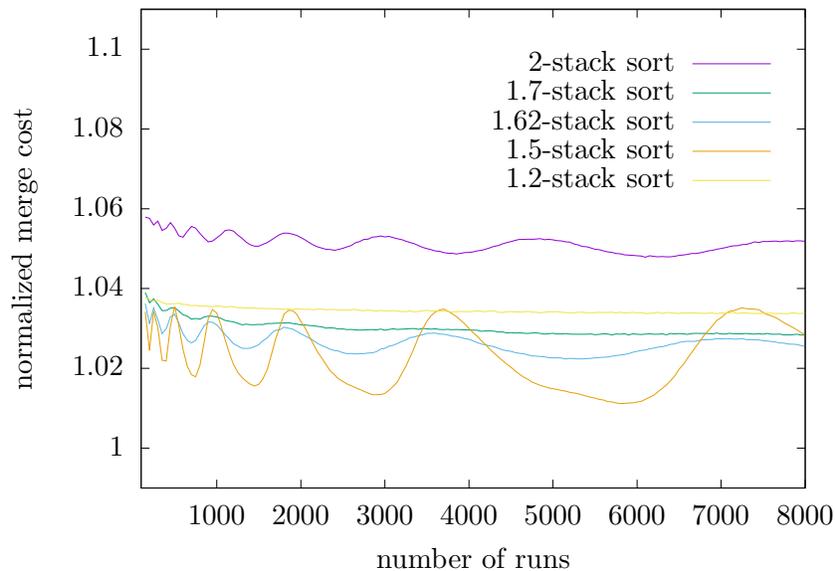}
    \caption{Comparison between $\alpha$-stack sorts for different $\alpha$ on
    uniform distribution over integers between 1 and 100.}
    \label{figure:alpha-merge-sort-plot}
\end{figure}

\begin{figure}
    \centering
    \input{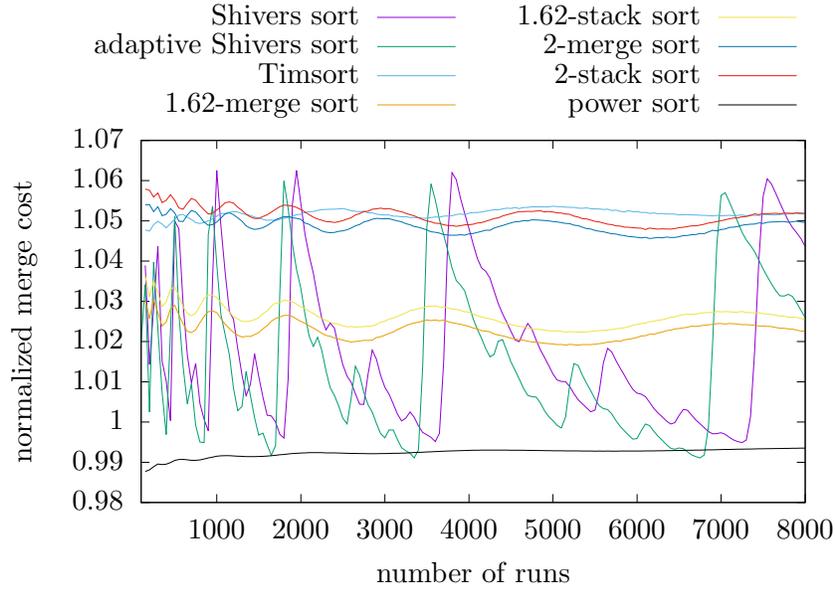}
    \caption{Comparison between sorting algorithms using the
    uniform distribution over integers between 1 and 100.}
    \label{figure:allsortsUniform}
\end{figure}

\begin{figure}
    \centering
    \input{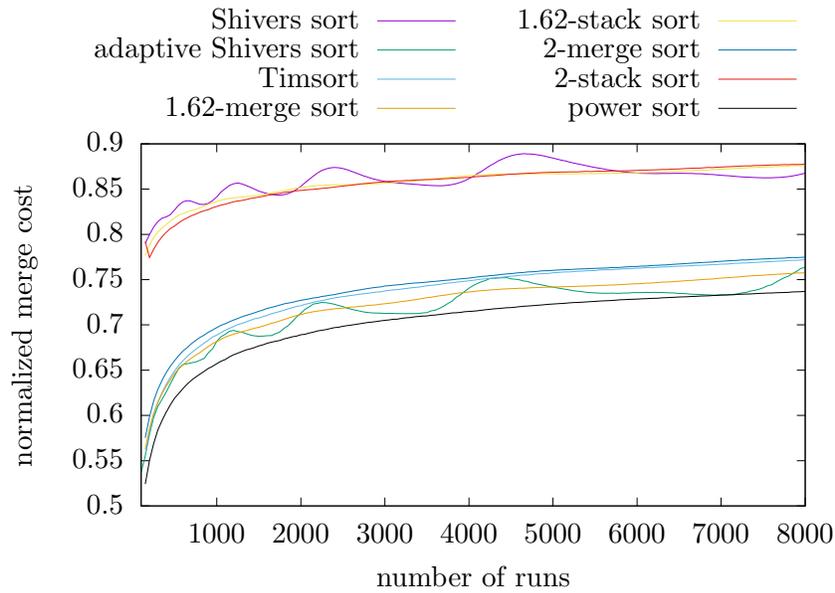}
    \caption{Comparison between sorting algorithms using a
            mixture of the uniform distribution over integers
            between 1 and 100 and the uniform distribution over integers between
            10000 and 100000, with mixture weights 0.95 and 0.05}
    \label{figure:allsortsMixed}
\end{figure}

\begin{figure}
\centering
\input{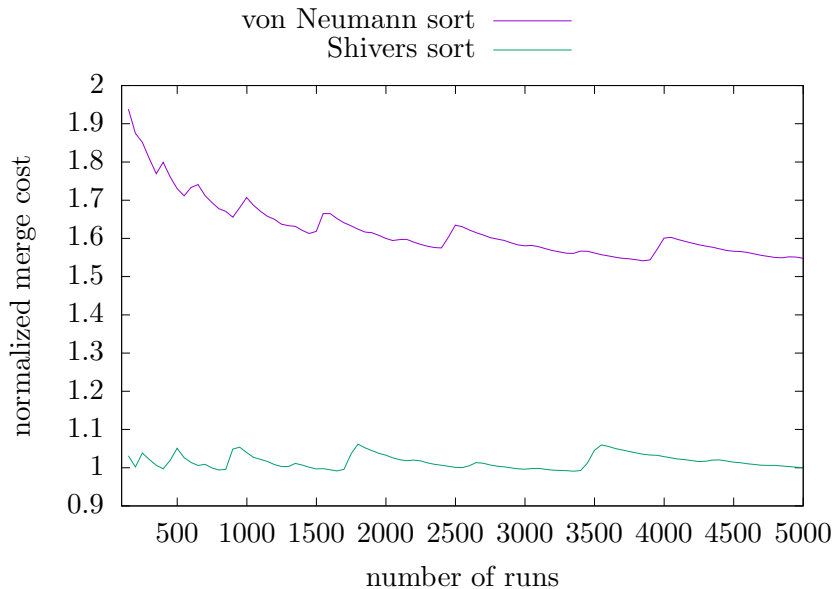}
\caption{The performance of the von Neumann merge sort (nonadaptive
merge sort) on the uniform distribution of run lengths
between 1 and 100. The random run lengths affect the total
input length~$n$, but the von Neumann sort ignores the runs in the initial runs in
the input sequence, and greedily builds binary trees of merges.
Unlike the other graphs, the data points on this graph are the result
of averaging 10 tests.}
\label{figure:vonNeumannUnif}
\end{figure}

Next we compared all the stable merge sorts discussed in
the paper, plus the adaptive Shivers sort of Jug\'e~\cite{Juge:AdaptiveShivers} and
the powersort of Munro-Wild~\cite{MunroWild:MergeSort}.  Figure~\ref{figure:allsortsUniform} reports on comparisons
using the uniform distribution.  It shows
that $1.62$-merge sort performs slightly better
than $1.62$-stack sort; and they perform better than Timsort,
$2$-stack sort and $2$-merge sort. The Shivers sort and the adaptive
Shivers sort performed
comparably to the $1.62$-stack and $1.62$-merge sorts, but exhibited
a great deal of oscillation in performance. This is discussed below,
but this is
presumably due to the Shivers and adaptive Shivers sorts rounding to powers of two.
The powersort algorithm performed the best in this test. The powersort knows
the total length~$n$ of the input to be sorted; this allows it to make
decisions on merges that approximate better a binary merge tree
over $n$~inputs.

Figure~\ref{figure:allsortsMixed} considers the mixed distribution. This is the
situation that is more meaningful for our stable merge sorts. As expected,
$1.62$-stack sort, $2$-stack sort, and Shivers sort performed the least well, approximately
15\% or 20\% worse than the other algorithms.
The other algorithms' performances were closer to each other. The Timsort
and $2$-merge sort were
marginally worse than the remaining three algorithms.  The
$1.62$-merge sort and the adaptive Shivers performed comparably, with
the adaptive Shivers sort exhibiting oscillatory behavior.
The powersort performed the best, presumably because its use of
the value~$n$ allowed it to make merge decisions that were
globally optimal.

The figures show that in many cases the merge cost oscillates periodically with the
number of runs; this was true in many other test cases not reported here. This seems
to be due in part to the fact that some of the sorts with this behavior (notably Shivers sort and
adaptive Shivers sort) base their merge decisions using powers of two, and do not take into
account the number $n$ of elements being sorted. In fact, the same behavior
is exhibited by the simple von Neumann sort, namely the non-adaptive sort that greedily
builds binary merge trees.  This is seen in Figure~\ref{figure:vonNeumannUnif}
using the uniform distribution. In this figure, for each value of~$m$,
the input sequences were chosen under the uniform distribution; these
have a fixed number, $m$, of runs, but different total input lengths~$n$.
That is, for fixed $m$,
the input length $n$ is a random variable;he von Neumann sort
ignores runs, and its merge cost depends only on~$n$.
The oscillations in the von Neumann and
Shivers sorts are clearly very similar, so we expect they have similar causes.
Under the mixed distribution (not shown), the oscillation with the von Neumann
sort was not present, presumably to the larger relative variance of~$n$. It is still present
for the Shivers and adaptive Shivers sort in the mixed distribution, presumably because
of their sensitivity to longer runs.

\section{Conclusion and open questions}\label{sec:conclude}
Theorem~\ref{thm:alphamergeUpper} analyzed
$\alpha$-merge sort only for $\alpha > \varphi$.
This leaves several open questions:
\begin{question}
    For $\alpha \le \varphi$, does $\alpha$-merge sort run in time
    $c_{\alpha} (1  +   \omega_m(1)) n \log m$?
\end{question}
\noindent
It is likely that when $\alpha<\varphi$, $\alpha$-merge sort
could be improved by making it 4-aware, or more generally,
as $\alpha\rightarrow 1$, making it $k$-aware
for even larger $k$'s.
\begin{question}
Is it necessary that the constants $d_\alpha \rightarrow \infty$
as $\alpha$ approaches~$\varphi$?
\end{question}

An \textit{augmented Shivers sort} can defined by
replacing the inner while loop on
lines \ref{algline:shiversYZ}-\ref{algline:ShiversInnerLoopB}
of the Shivers sort
with the code:

\noindent
\hspace*{0.5in}\hbox{
    \vbox{
        \begin{algorithmic}
            \While{$2^{\lfloor{\log |Y|}\rfloor} \le |Z|$}
              \If {$|Z|\le|X|$}
                \State{Merge $Y$ and $Z$}
              \Else
                \State{Merge $X$ and $Y$}
              \EndIf
            \EndWhile
        \end{algorithmic}
    }
}
\noindent
The idea is to incorporate the $3$-aware features of $2$-merge sort
into the Shivers sort method. The hope is that this might give an
improved worst-case upper bound:
\begin{question}
    Does the augmented Shivers sort run in time $O(n \log m)$?
    Does it run in time $(1 +  o_m(1)) n \log m$?
\end{question}
\noindent
The notation $o_m(1)$ is intended to denote a value that
tends to zero as $m\rightarrow\infty$.

Jug\'e's adaptive Shivers sort~\cite{Juge:AdaptiveShivers} is similar
to the above augmented Shivers sort, and it does have run time
$(1 +  o_m(1)) n \log m$. In fact it is asymptotically optimal
in a rather strong sense: It
has run time asymptotically equal to the entropy-based lower
bounds of Barbay and Navarro~\cite{BarbayNavarro:AdaptiveSorting}.
This also answered the open Question~4 in~\cite{BussKnop:StableMergeSortSODA}
about the existence of asymptotically optimal $k$-aware algorithms.
The powersort algorithm of Munro and Wild~\cite{MunroWild:MergeSort}
has similar asymptotic run time bounds.

Our experiments indicate that $1.62$-merge sort, the adaptive
Shivers sort, and powersort are all strong contenders for
the ``best'' adaptive merge sort in practice.
The powersort was the best in our experiments by a
small margin, but it uses a more complicated algorithm
to decide when to merge, and this will increase the running time.
(It potentially will use $O(\log n)$ time
per merge decision, although this can be improved on.) The adaptive
Shivers sort has better theoretical worst-case upper bounds than
$1.62$-merge sort, but this probably only arises in rare cases and it
is a small difference in any event.
The adaptive Shivers sort also exhibited some oscillation.
Both of these algorithms are easy to implement;
neither has
a clear edge in our experiments.

Our experiments used random distributions that probably
do not do a very good job of modeling real-world data.
Is it possible to create better models for real-world data?
Finally, it might be beneficial to run experiments with actual real-world data,
e.g., by modifying deployed sorting algorithms to calculate statistics based
on run lengths that arise in actual practice.

\bibliographystyle{siam}
\bibliography{logic}

\end{document}